\documentclass[a4paper,12pt]{article}

\usepackage{cite}
\usepackage{fullpage,setspace,hyperref}

\usepackage{amssymb,amstext,amsmath,amsthm,amsfonts,amsmath,amscd,relsize,url,mathtools, bm}
\usepackage{accents}
\usepackage[utf8]{inputenc}
\usepackage{bbm} 

\usepackage{color}

\allowdisplaybreaks
\numberwithin{equation}{section}

\newcommand{\be}{\begin{equation}}
\newcommand{\ee}{\end{equation}}
\newcommand{\bea}{\begin{eqnarray}}
\newcommand{\eea}{\end{eqnarray}}

\def\pa{\partial}

\newcommand{\p}{\partial}
\newcommand{\pt}{\tilde{\partial}}
\newcommand{\n}{\nabla}
\newcommand{\rd}{\mathrm{d}}

\def\to{\rightarrow}

\def\tl{\tilde}
\def\tx{\tilde{x}}
\def\ty{\tilde{y}}

\def\gld{generalized Lie derivative }
\def\glds{generalized Lie derivatives }

\def\cyc{\sum_{Cycl\,(X,Y,Z)}}

\newcommand{\PS}{\mathcal{P}}

\newcommand{\TT}{\mathbbm{T}}
\newcommand{\Lie}{\mathcal{L}}

\newcommand{\ap}{\alpha}
\newcommand{\bt}{\beta}
\newcommand{\Db}{\mathrm{D}}
\newcommand{\Dbr}{\mathrm{D}}
\newcommand{\pr}{P}

\newcommand{\PP}{P}
\newcommand{\PPt}{\tilde{P}}
\newcommand{\id}{\mathbbm{1}}

\newcommand{\se}{\Gamma}
\newcommand{\bl}{[\![}
\newcommand{\br}{]\!]}

\newcommand{\lc}{\mathring{\n}}

\newcommand{\Lb}{\mathbf{L}}
\newcommand{\im}{\mathrm{Im}\,}
\newcommand{\Ker}{\mathrm{Ker}\,}
\newcommand{\Hh}{\mathcal{H}}
\newcommand{\HH}{\mathcal{H}}

\newcommand{\Lt}{\tl{L}}
\newcommand{\etah}{\hat{\eta}}
\newcommand{\xt}{\tl{x}}
\newcommand{\yt}{\tl{y}}
\newcommand{\zt}{\tl{z}}

\newtheorem{Thm}{Theorem}
\newtheorem{Lem}{Lemma}
\newtheorem{Cor}{Corollary}
\newtheorem{Prop}{Proposition}
\newtheorem{Def}{Definition}

\newtheorem*{Rem}{Remark}

\begin{document}

\begin{titlepage}
\vfill

\begin{flushright}
LMU-ASC 35/17 \\
\end{flushright}

\vfill
\begin{center}
   \baselineskip=16pt
   	{\Large \bf 
	%Generalizing Generalized Geometry  and \\
       Generalised Kinematics for Double Field Theory}
   	\vskip 2cm
   	{\sc  Laurent Freidel$^*$\footnote{\tt lfreidel@perimeterinstitute.ca}, 
          Felix J. Rudolph$^{\dagger}$\footnote{\tt felix.rudolph@lmu.de}, 
          David Svoboda$^*$\footnote{\tt dsvoboda@perimeterinstitute.ca} }  
	\vskip .6cm
    {\small  \it 
        $*$ Perimeter Institute for Theoretical Physics, \\
          31 Caroline St. N.,  Waterloo ON, N2L 2Y5, Canada \\ \ \\
        $\dagger$ Arnold Sommerfeld Center for Theoretical Physics, Department f\"ur Physik, \\
          Ludwig-Maximilians-Universit\"at M\"unchen, Theresienstr. 37, 80333 M\"unchen, Germany }
	\vskip 2cm
\end{center}

\begin{abstract}
We formulate a  kinematical extension of Double Field Theory  on a $2d$-dimen-sional para-Hermitian manifold $(\PS,\eta,\omega)$ where the $O(d,d)$ metric $\eta$ is supplemented by an almost symplectic two-form $\omega$. Together $\eta$ and $\omega$ define an almost bi-Lagrangian structure $K$ which provides a splitting of the tangent bundle $T\PS=L\oplus\Lt$ into two Lagrangian subspaces. In this paper a canonical connection and a corresponding generalised Lie derivative for the Leibniz algebroid on $T\PS$ are constructed. We find integrability conditions under which the symmetry algebra closes for general $\eta$ and $\omega$, even if they are not flat and constant. This formalism thus provides a generalisation of the kinematical structure of Double Field Theory. We also show that this formalism allows one to reconcile and unify Double Field Theory with Generalised Geometry which is thoroughly discussed.

\end{abstract}

\vfill

\setcounter{footnote}{0}
\end{titlepage}

\tableofcontents

\section{Introduction}
The notion of spacetime geometry used in the context of general relativity is based on four elements $(M,[\cdot,\cdot], g, \nabla)$: a set of points $M$ equipped with a differentiable structure, which is encoded into the Lie bracket on the tangent bundle $[\cdot,\cdot]:\se(\Lambda^2(TM)) \to \se(TM)$; a metric $g\in \se(T^*M\otimes T^*M)$, which gives us a notion of distance between points; and a compatible torsionless connection $\nabla: \se(TM) \to \se(TM\otimes T^*M)$ which defines a notion of parallel transport. The data $(M,[\cdot,\cdot])$ which encodes the underlying structure of spacetime and the symmetry of the theory by infinitesimal diffeomorphisms, will be referred to as the {\it kinematical} component of the geometry, since in the usual formulation of gravity it is given from the onset and is not subject to any dynamics. The data $(g,\nabla)$ will be referred to as the {\it dynamical} component; in the Einsteinian picture, this component is subject to the dynamics encoded into the choice of an action invariant under diffeomorphisms that fixes the torsion to be zero and determines the Ricci component of the curvature. 

By the generalisation of geometry we mean a structure that generalises {\it both}  its  kinematics $(M,[\cdot,\cdot])$ and dynamics $(g,\nabla)$. The purpose of this paper is to focus on the generalisation of the kinematics while a companion paper \cite{toappear} focuses on the generalisation of the dynamics.  
As we will see, one of the main components of  our results is that the kinematics itself when properly generalised is encoded into geometrical elements denoted by $(\eta,\omega)$. 
This opens up  a radically  new  possibility: the kinematical structure of spacetime itself might be subject to a dynamical selection rule. 

The construction presented here is profoundly motivated by the study of effective string geometries. 
It is now well established that the arena of several novel formulations of string theory and quantum gravity is a mathematical space denoted by $\PS$  which is locally given by $\PS \simeq M\times \tilde{M}$ with spacetime denoted by $M$. This is the  doubled space of Double Field Theory (DFT) \cite{Tseytlin:1990nb, Tseytlin:1990va, Siegel:1993xq,Siegel:1993th,Hull:2009mi}, it  has twice the number of dimensions of spacetime  and it is seen as a target space for the string, its coordinates are a combination of spacetime and winding coordinates. This new picture has in particular allowed to deepen our understanding of flux compactifications \cite{Aldazabal:2011nj,Geissbuhler:2011mx,Andriot:2011uh,Andriot:2012an} and it has also quite remarkably led us to an understanding of  the constraints of supergravity purely from geometrical argument and without the need of supersymmetry \cite{Coimbra:2011nw,Hohm:2013vpa,Aldazabal:2013sca}. Recently, the correspondence between the string target and the effective description has been revisited in the context of a duality symmetric formulation of string theory called metastring theory. A key element of this description has been to establish that the doubled space $\PS$ of DFT is in fact a phase space carrying a symplectic structure \cite{Freidel:2013zga,Freidel:2014qna,Freidel:2015pka}. The presence of this structure radically modifies the interpretation of the effective string geometry \cite{Freidel:2015uug}. This has been used to show that at the quantum level the string target can be thought of as a non-commutative manifold \cite{Freidel:2017wst}.

The key geometrical elements entering the construction of DFT \cite{Hohm:2010pp,Hohm:2012mf} are a {\it flat} $O(d,d)$ metric\footnote{An $O(d,d)$ metric is a neutral metric with split signature $(+\cdots+,-\cdots-)$.} denoted $\eta$ and a {\it generalised} metric $\HH$ which provides an $O(d)\times O(d)$ structure on the generalised tangent bundle in the Euclidean case, or $O(1,d-1)\times O(1,d-1)$ in the Lorentzian case. The metric $\eta$ encodes the kinematical structure of DFT while $\HH$ encodes its dynamics. These two geometrical structures satisfy a compatibility condition which follows from the T-duality symmetry and imposes that $\HH$ is an element of the coset space $O(d,d)/O(d)\times O(d)$.  In order to write this relation, let us recall that we can think of $(\eta,\HH)$
% \in T^*M \times T^*M$ 
as $(2,0)$ tensors  or as maps $T\PS\rightarrow T^*\PS$, in which case we denote them by $(\hat{\eta},\hat{\HH})$.
%: TM \to T^*M$.
Defining a {\it chiral structure} $J:=  \hat{\eta}^{-1}\hat{\Hh}\in \mathrm{End}(T\PS)$, the compatibility between $\eta$ and $\Hh$ is expressed as the condition
\be
J^2=+\id, \qquad J^T\hat{\eta}J = \hat{\eta} \, .
\ee

One of the central ingredients of DFT is the idea that the notion of diffeomorphism symmetries acting on the spacetime $M$ is unified with the $B$-field gauge transformations  into a generalisation of diffeomorphisms 
 acting on the doubled space $\PS$. At the infinitesimal level this action is given by  the Dorfman derivative $\Lb_X$. This is a generalisation of the notion of  Lie derivative -- labeled by a vector $X \in \se(T\PS)$ -- which acts on tensors in $\PS$ \cite{Hull:2009zb}. By construction, this generalised Lie derivative preserves the neutral metric: $\Lb_X\eta=0$.  

Unlike the usual Lie derivative, the Dorfman derivative defines an infinitesimal transformation that does not integrate to a group action in general. In order to overcome this limitation, the strategy used in double field theory is to restrict the space of vector fields and tensors on which the generalised Lie derivative acts, and demand that they satisfy the so-called {\it section condition }. In the simplest cases the section condition imposes that $\Box_{\eta} (\Phi)=\Box_{\eta} (\Phi \Phi')=0$ for all fields and their products, where $\Box_{\eta}=\eta^{AB}\pa_A \pa_B$.  When  acting on tensor fields that satisfy the section condition, and for vector fields that also satisfy the section condition, it can be shown that the Dorfman derivative satisfies the Jacobi identity $[\Lb_X,\Lb_Y] \Phi = \Lb_{ \bl X ,Y \br} \Phi$  (where the bracket $\bl\cdot,\cdot\br$ is the Courant bracket obtain by skew-symmetrisation of the Dorfman derivative) and in this case it therefore defines a notion of symmetry. This symmetry can in turn be used to constrain the form of  stringy gravitational actions without the need of supersymmetry \cite{Coimbra:2011nw,Hohm:2013vpa,Aldazabal:2013sca}. One of the main challenges of this approach is to find a consistent mathematical description of the section condition that translates directly into an underlying choice of a geometrical structure.
% Indeed the section condition implies that the fields should effectively depends on half the dimensions of $\PS$, which half is left undetermined. 

A separate but related development -- also aiming at describing geometrical entities relevant to the geometry of string theory -- is  an approach that does not have the previous issue; this  is the field of Generalised Geometry (GG). This field originates from  the attempts of Courant, Weinstein \cite{Courant-Weinstein} and Dorfman \cite{Dorfman} to unify the Dirac formalism \cite{Dirac} of symplectic reduction with the formalism of  Poisson reduction by first and second class constraints. This was later also studied by Courant \cite{courant1990dirac}, who introduced a bracket, which is the skew-symmetrisation of the Dorfman Lie derivative. The construction was subsequently generalised by Liu, Weinstein and Xu \cite{Liu-Weinstein}. This led to the construction of the notion of a Leibniz algebroid\footnote{Equivalently known as Courant algebroid.}\cite{Roytenberg}.

In the approach of generalised geometry, the spacetime manifold $M$ is not doubled and instead the tangent bundle over $M$ is generalised. In the simplest case (see \cite{Baraglia, Xu} for more general examples) it is promoted to the double tangent bundle $\mathbb{T}M := TM \oplus T^*M$. The advantage of this approach is that there is no need to impose a section condition since the base manifold is unchanged to begin with. The generalised tangent bundle naturally carries a symmetric pairing $\langle\cdot, \cdot\rangle_+$ which descends from the canonical pairing of forms and vectors
\begin{equation}
\langle (x,\ap), (y,\bt)\rangle_+ : =\ap(y)+\bt(x),\qquad x,y\in \se(TM), \,\, \alpha,\beta \in \se(T^*M) \,.
\end{equation}
This symmetric pairing of GG plays a role analogous to the flat metric $\eta$ of DFT. This bundle $\TT M$ also carries a Dorfman derivative that generalises the Lie derivative and is defined by 
\be
 \mathbb{L}_{(x,\alpha)} (y,\beta) = ( [x,y], {\cal L}_x \beta - i_y \rd \alpha).
\ee
This derivative generates a group of transformations without the need of additional constraints.
This kinematical structure can then be augmented by a generalised metric $\HH$ which is a pairing on $\mathbb{T}M$.  %The definition and the construction of the corresponding connection compatible with $\HH$ is going to be the focus of the our next paper \cite{toappear}.

Despite the many successes of both the DFT and GG approaches, there are still some unsatisfactory aspects that motivate us to go beyond the state of the art. Starting from DFT, it has been noted that locally the section condition forces the fields to depend on only half the coordinates. However, there is not a unique solution to this constraint, i.e. there is no geometrical information that reveals which half of the coordinates the fields depend on. In other words, which submanifold of $\PS$ serves as the base $M$ for the generalised geometry $\TT M$. This freedom of choice makes DFT in a sense undetermined.
 
This would suggest that GG is therefore the proper mathematical expression of string geometry since it does not contain any undeterminations. However, from the point of view of string theory, the main reason to expect a generalisation of geometry is the T-duality symmetry. The stringy expression of T-duality in terms of sigma models suggests that T-duality is a map from $M$ (which appears in the large $R$ limit of compactified strings) onto a dual manifold $\tilde{M}$ (which appears in the small $R$ limit of compactified strings), i.e. a symmetry that exchanges the base of the generalised bundle. The action of changing the base manifold can eventually be implemented  in DFT since $M$ is viewed as a subspace (or a quotient space) of $\PS$, but it cannot be accommodated in GG since there the base is fixed. At best  we can implement an $O(d,d)$ transformation on the fibers $\mathbb{T}M$ in lieu of a T-duality transformation, but this does not express the profound change of geometrical structure inherent to T-duality symmetry\footnote{We now understand that this change requires non-commutative geometry \cite{Freidel:2017wst,Freidel:2017nhg}.}.

To summarise, we see that DFT can accommodate some aspects of T-duality symmetry but is mathematically undetermined, while GG is mathematically sound but does not provide a geometrically satisfying description of T-duality. It is therefore of utmost importance to understand what extra geometrical data enables generalised geometry to be embedded inside double field theory in order to eventually allow a mathematical characterisation of T-duality in the context of GG.

One of the key ingredients in the resolution of this puzzle has been revealed in the study of metastring theory \cite{Freidel:2014qna,Freidel:2015pka} and in the corresponding modular spacetime that represents its effective geometry \cite{Freidel:2015uug}. The metastring theory is a formulation of string theory in which T-duality symmetry is manifest. The central question that has been investigated in these works is what type of generalisation of target space geometry allows to encapsulate 
stringy phenomena. It has been found that the proper geometrical concept generalising the notion of spacetime is the idea of modular spacetime which is fundamentally quantum \cite{Freidel:2015uug,Freidel:2016pls}. At the classical level, this translates into the fact that modular spacetimes can be viewed as a non-trivial line bundle over a {\it Born geometry} \cite{Freidel:2013zga} that constitutes the classical geometrical expression of the quantum string geometry.

The main point -- which is in total agreement with DFT -- is that the string target space is a doubled space. The key element that differentiates the Born geometry from the doubled space of DFT is the fact that this target space should be interpreted at the classical level as a dynamical {\it  phase space} of string probes combining spacetime and energy-momentum space. The presence of a dynamical symplectic form $\omega$ in addition to the DFT geometrical elements $(\eta, \HH)$ completes the picture and allows one to connect the generalised geometrical concepts with concepts in non-commutative geometry \cite{Freidel:2017wst}.

This means that the kinematics of the metastring target space $\PS$ is encoded in the pair $(\eta, \omega)$, where $\eta$ is a neutral metric and $\omega$ is a non-degenerate two-form. These satisfy a compatibility condition which can be stated as follows. The operator $K:= \hat{\eta}^{-1}\omega\in \mathrm{End}( T\PS)$ defines an {\it almost product structure} compatible with $\eta$, i.e.
\be
K^2 = +\id, \qquad K^T \hat{\eta} K = - \hat{\eta}. 
\ee
In other words, the data $(\PS,\eta,\omega)$ is an almost para-Hermitian manifold.  As shown in \cite{Freidel:2017wst} this para-Hermitian structure is essential in order to construct the string vertex operators. This implies that the tangent bundle decomposes as $T\PS = L\oplus \tilde{L}$ where $L$ is a distribution defined to be the eigenspace of $K$ with eigenvalue $+1$ and $\tilde{L}$ is the complementary distribution with eigenvalue $-1$. Both $L$ and $\tilde{L}$ are {\it null}, $\eta|_L=0$ and {\it Lagrangian}, $\omega|_L=0$.

The key difference between the metastring target space and the DFT realm is twofold. First, in Born geometry we do not  assume $\eta$ to be flat\footnote{For a different approach to DFT on curved spaces, namely on group manifolds, see \cite{Blumenhagen:2014gva,Hassler:2016srl}.} and we also do not assume $\omega$ to be closed. Moreover, the compatible two-form $\omega$ and the associated Lagrangians $L$ and $\tilde{L}$ provide exactly the extra geometrical structure needed to implement the section condition geometrically. In other words, promoting $\PS$ to a para-Hermitian manifold is what is needed in order to generalise the kinematical structure of DFT and formalise the section condition. This is what we intend to prove.

Accordingly, the goal of this paper is also twofold: on one hand, we are going to show that we can generalise the notion of the Dorfman derivative even when $\eta$ is curved. The only condition that we will impose is the integrability of the distribution $L$. On the other hand, we are going to show that the extra para-Hermitian structure $\omega$ is exactly what is needed in order to connect the kinematical part of Born geometry to an extension of generalised geometry. In other words, we are going to describe how to unify DFT with GG.
In summary, the key results of this paper are: 
\begin{itemize}
	\item The proposition that para-Hermitian geometry encoded into the pair $(\eta,\omega)$ 	and satisfying a new type of closure and integrability conditions  is the proper geometrical structure that allows  to formulate DFT on a manifold with curved $\eta$ (cf. \cite{Cederwall:2014kxa}). We also include key examples along the way and show  that the usual DFT naturally lives on a flat para-K\"ahler manifold. 
	\item A generalisation of the Dorfman derivative for curved $\eta$ in DFT based on a canonical para-Hermitian connection that generalises the Levi-Civita connection  and a proof that it satisfies the Jacobi identity under our closure and integrability conditions. 
 	\item An extension of the realm of generalised geometry to foliations associated with the para-Hermitian geometry  that allows us to construct an isomorphism between the DFT and GG frameworks. Under this isomorphism the generalised Lie derivatives of DFT and GG can be identified.
\end{itemize}
 
Most of the works on GG and DFT have so far not considered the presence of the extra almost symplectic structure $\omega$ -- which gives the more reduced structure group $GL(d)$ as we will see -- and the key geometrical role it plays\footnote{More general approaches to graded symplectic manifolds in the context of doubled spaces can be found in \cite{Heller:2016abk,Deser:2016qkw}.}. One notable exception is in the work of Hull \cite{Hull:2004in} who formalised the geometry of non-geometric vacua in terms of a flat Born geometry. The other important exception is the work of Vaisman \cite{Vaisman:2012ke,Vaisman:2012px} who was the first to recognise that the relationship between GG and DFT hinges on a para-Hermitian or a para-K\"ahler structure. We would also like to draw attention to the earlier work of Alvarez \cite{Alvarez:2000bh,Alvarez:2000bi} on target space duality in the context of symplectic manifolds. Our work can be viewed as an extension of these pioneering works profoundly inspired by the metastring formulation of string theory \cite{Freidel:2015pka}.

The remainder of this paper is organised as follows: Section 2 provides a concise review of the generalised geometry and double field theory approaches. In Section 3 the almost symplectic form $\omega$ is introduced to give a formulation in terms of almost para-Hermitian or almost para-K\"ahler manifolds. A discussion on integrability is also included and the concepts of $L$-para-Hermitian and $L$-para-K\"ahler manifolds are introduced. In Section 4 the generalisation of the kinematics of the doubled space is presented by constructing a generalisation of the Dorfman derivative and finding suitable connection on the para-Hermitian manifold. This section contains our main theorem which states that our generalisation of the Dorfman derivative satisfies the Jacobi identity on an $L$-para-Hermitian manifold.  In Section 5 the relation of this construction to generalised geometry is made precise by constructing an explicit isomorphism between them and showing that this isomorphism identifies the two notions of generalised Lie derivatives. Lastly, the appendix contains technical details and computations to supplement the main text.

\section{Lightning Reviews of GG and DFT Kinematics}
\label{sec:review}
Before embarking onto an extension of generalised geometry and double field theory  to the case when $\eta$ is curved, we will review the main geometrical elements of both GG and DFT, restricting our summary to the kinematical aspects. We then outline how the two (a priori different) setups can be related in the proposed framework. More details on this interplay formulated in the language of  {\it algebroid} structures will be given in \cite{inprogress}.

\subsection{Review of the Generalised Geometry Setup}
\label{sec:GG}
In generalised geometry the central object of study is the doubled bundle $\mathbb{T}M :=  TM \oplus T^*M$  over a manifold $M$. This bundle is equipped with a symmetric pairing described as follows: given two sections $ X= (x,\alpha)$ and $Y=(y, \beta)$ of $\TT M$, it is given by
 \be
 \eta(X,Y) = \iota_x \beta  + \iota_y\alpha, 
 \ee
where $\iota_x$ denotes the interior product. The doubled bundle $\TT M$ is also equipped with a natural generalisation of the Lie derivative of $TM$, the so-called Dorfman bracket or \gld 
 \be
 \mathbb{L}_{(x,\alpha)}(y,\beta) = ([x,y],{\cal L}_x \beta - \iota_y\rd\alpha),
 \label{eq:genLieGG}
 \ee
 where ${\cal L}_x$ denotes the usual Lie derivative along $x$. There is also a natural projection $\pi: \TT M \rightarrow TM$ given by $\pi(x,\ap)=x$. Using this, we find that the three structures $(\eta,\pi,\mathbb{L})$ interplay in a natural way. First, $\mathbb{L}$ is compatible with $\eta$ through $\pi$
\begin{equation}
\pi(z,\gamma)[\eta((x,\ap),(y,\bt))]=\eta(\mathbb{L}_{(z,\gamma)}(x,\alpha),(y,\bt)) + \eta((x,\alpha),\mathbb{L}_{(z,\gamma)}(y,\bt)) \, .
\end{equation}
Second, $\mathbb{L}$ and $\pi$ satisfy the anchoring property
\begin{equation}
\pi( \mathbb{L}_{(x,\alpha)} {(y,\beta)}) = {\cal L}_{x}y=[\pi(x,\ap),\pi(y,\bt)] \, .
\end{equation}
Third, $\mathbb{L}$ is normalised with respect to $\eta$
\begin{equation}
\eta( \mathbb{L}_{(x,\alpha)} (x,\alpha),(y,\beta)) = \tfrac12 \pi(y,\beta) [\eta((x,\alpha),(x,\alpha))] \, .
\end{equation}
In addition, $\mathbb{L}$ satisfies the Jacobi identity
 \be
 [\mathbb{L}_{(x,\alpha)},\mathbb{L}_{(y,\beta)}]=\mathbb{L}_{\mathbb{L}_{(x,\alpha)}{(y,\beta)}}. 
 \label{eq:GGJac}
 \ee
The four properties above define a so-called {\it Leibniz} (or {\it Courant}) algebroid on $\TT M$.

\subsection{Review of the Double Field Theory Setup}

In double field theory we start with a C-bracket defined on the doubled space $\PS = M \times \tilde{M}$
\begin{equation}
\bl X,Y\br^{\pa} \coloneqq \partial_XY - \partial_YX+\frac{1}{2}(\theta(Y,X)-\theta(X,Y)),
\label{eq:DFTcbracket}
\end{equation}
where we have introduced the vector $\theta^A(X,Y):=\eta^{AB}\eta(X,\pa_B  Y)$\footnote{In the $O(d,d)$ index notation used in DFT, the C-bracket reads
$$(\bl X,Y\br^{\pa})^A =  X^B(\pa_BY)^A-Y^B(\pa_BX)^A + \frac{1}{2}\eta^{AB}\eta_{CD}\left((\pa_BX)^CY^D-(\pa_BY)^CX^D\right).$$ The vector $\theta^A(X,Y)$ captures the difference to the ordinary Lie bracket and is therefore expressible in terms of the Y-tensor of DFT \cite{Berman:2012vc}: $\theta^A(X,Y) = {Y^{AB}}_{CD}X^C\partial_BY^D$.}. The tangent bundle $T\PS$ splits as $T\PS=TM\oplus T\tilde{M}\coloneqq L\oplus \Lt$. %This geometry is called a product manifold [].

The  space $\PS$  is equipped with a {\it flat} non-degenerate metric $\eta$ such that the distribution $L$ and $\tilde{L}$ are isotropic with respect to $\eta$. If $(x^\mu,\tilde{x}_\mu)$ are local coordinates of $\PS$, such that $x^\mu$ and $\xt_\mu$ are coordinates on $M$ and $\tilde{M}$, respectively, $\eta$ can be written as 
\be
\eta = \rd x^\mu \otimes \rd \tilde{x}_\mu + \rd \tilde{x}_\mu \otimes \rd x^\mu.
\ee
The transition functions between coordinate patches that preserve these flat coordinates belong to $O(d,d)$. The generalised Lie derivative, whose skew-symmetrization gives the C-bracket, is given by
\be \label{eq:genlieDFT}
\Lb^\pa_XY = \partial_XY - \partial_YX + \theta(Y,X).
\ee
This \gld $\Lb^\pa_XY$ does not, however, satisfy the Jacobi identity. The Jacobiator capturing this failure is defined for any generalised Lie derivative $\Lb$ by (cf. Definition \ref{def:jacobiator} and Appendix \ref{sec:jacobiator})
\begin{equation}
J(X,Y,Z,W) = \eta( [\Lb_X,\Lb_Y] Z, W) - \eta( \Lb_{\Lb_XY} Z, W),
\end{equation}
and an explicit computation shows that for the \gld $\Lb^\p$ in \eqref{eq:genlieDFT} we have
  \bea\label{eq:jacobiatorDFT}
J^\pa(X,Y,Z,W) =  \eta( \theta(Z,X),\theta(W,Y))
- \eta(\theta(Z,Y),\theta(W,X))-
\eta(\theta(Y,X),\theta(W,Z)) \, .
  \eea

The usual way to satisfy the Jacobi identity \cite{Hull:2009mi, Aldazabal:2013sca} in double field theory is to {\it restrict} the vector fields by imposing the {\it section condition}, which implies $\eta^{AB}\p_A\p_B=0$ and hence $\eta_{AB}\theta^A\theta^B=0$ and $J^\p=0$. This construction is not canonical and from a geometrical perspective rather ad-hoc.
 
A more fruitful approach is to impose the restriction on the type of generalised Lie derivative one allows. Let us first adapt our notation by generalising \eqref{eq:genlieDFT} to any metrical connection $\n$ as first developed by Vaisman \cite{Vaisman:2012ke},
\be
\Lb^\n_XY = \n_XY - \n_YX + \theta_\n(Y,X),
\ee
where now $\theta_\n(Y,X)$ is defined such that $\eta(Z,\theta_\n(Y,X)) := \eta(Y,\n_ZX)$. Given two complementary projectors $\PP,\PPt: T\PS \to T\PS $, such that $\im \PP$ and $\im \PPt$ are maximally isotropic with respect to $\eta$, the tangent bundle splits as $T\PS=\im \PP \oplus \im \PPt$. We can then decompose the covariant derivative as $\n = \Dbr + \tilde{\Dbr}$ where $\Dbr_X := \n_{\PP(X)}$ and $\tilde{\Dbr}_X :=\n_{\PPt(X)}$. Accordingly, this gives a decomposition of the generalised Lie derivative in terms of the projected Lie derivatives 
  \be \label{eq:splitgenlie}
  \Lb_X^{\n}Y = \Lb^{\Dbr}_XY + \Lb^{\tilde{\Dbr}}_XY.
  \ee 
Coming back to the DFT setting, where $\n=\p$, one can check that the Jacobiators for the projected \glds $J^\Dbr$ and $J^{\tilde{\Dbr}}$ identically vanish. For $J^\Dbr$, this is a result of the following facts:
\begin{itemize}
\item the curvature of $\pa$ vanishes
\item the map $\PP$ is an anchor of $\Lb^{\Dbr}$, i.e. $ \PP(\Lb^{\Dbr}_XY)= [\PP(X),\PP(Y)]$
\item the vector $\theta_D$ belongs to $\Ker\PP$, i.e. $\PP(\theta_\Dbr)=0$.
\end{itemize}
Similar statements hold for $J^{\tilde{\Dbr}}$.

We can see that the splitting \eqref{eq:splitgenlie} corresponds in this setup to imposing the section condition. However, the fields are still unrestricted, but the projected \glds only ``see'' the dependence on half the coordinates. This approach is more geometrically appealing because of the fact that the fields remain unrestricted and different solutions of the section condition can still be recovered by choosing different splittings\footnote{The standard solution of the section condition where fields depend only on $x^\mu$ or only on $\xt_\mu$ corresponds to $\PP$ and $\PPt$ projecting to $TM$ and $T^*M$, respectively.} of $T\PS$.

\subsection{Relationship between GG and DFT}
We have seen above that while in generalized geometry one extends the fibre directions to get $\TT M$, in DFT the base is extended to give $\PS= M \times \tilde{M}$. Even though the bundles of interest, namely $\TT M$ and $T\PS$, have the same rank, the dimensions of the underlying manifolds do not match. To relate the two setups, one needs to find a way to ``reduce'' the base $\PS$ to get a half-dimensional base $M$. Intuitively, this could be seen as one of the reasons why the section condition of DFT needs to be imposed, but as we have explained in the previous discussion, this has to be spelled out more carefully.

Let us now consider the following construction. Starting with $T\PS$, the tangent bundle of a $2d$-dimensional manifold, we introduce a splitting into two rank $d$ distributions $L$ and $\Lt$ such that $T\PS=L\oplus\Lt$. If one of the distributions is integrable (say, $L$, without loss of generality), we get a foliation ${\cal F}$ of $\PS$, such that $T{\cal F}=L$ is the integrable subdistribution. $\cal F$ is a d-dimensional manifold which is bijectively identified with $\PS$ \cite{Milnor, Molino, FoliationConnection}. It  is given by the union of leaves ${\cal F} =\coprod_{[p]} M_{[p]}$ .  We can now view  $T\PS$ as a rank $2d$ bundle over the  $d$-dimensional foliation manifold $\cal F$, which provides us with a map $T\PS \to T{\cal F}$. 
Working in local coordinates, where $x^\mu$ are coordinates along the leaves $M_i$ and $\xt_\mu$ the {\it transverse} coordinates (i.e. coordinates that are constant along $M_i$), $T\PS$ is spanned by $\p_\mu=\frac{\p}{\p x^\mu}$ and $\pt^\mu=\frac{\p}{\p \xt_\mu}$ while $T{\cal F}$ is spanned by $\pa_\mu$. We can then construct an isomorphism between  $T\PS$ and $T {\cal F}$, which can be  expressed in local coordinates as
\begin{equation}\label{eq:local_iso}
\pt^\mu \leftrightarrow \rd x^\mu.
\end{equation}

In the usual DFT setting, $\PS$ is a product of affine manifolds. That is $P=M\times\tilde{M}$ where $M$ and $\tilde{M}$ are vector spaces or flat tori. % The previous considerations allow to map the DFT construction in the GG framework by restricting to a leaf of $\mathcal{F}$ and  the two approaches coincide.
The situation is considerably simplified since the coordinates of $\PS$ split globally into $(x^\mu,\xt_\mu)$. However, in the general setting of the Born geometry, we have to consider issues stemming from global considerations. For example, $\PS$ is not globally a product manifold $M\times \tilde{M}$, and the splitting is only done on the level of the tangent bundle $T\PS=L\oplus \Lt$. To recover the distinguished local coordinates $(x^\mu,\xt_\mu)$, we therefore have to impose the integrability of $L$ as described in the previous paragraph. It turns out, as first observed by Vaisman \cite{Vaisman:2012px}, that a natural framework for these considerations is given by {\it para-Hermitian geometry}, which is the approach we describe in the following sections.

\section{Para-Hermitian, para-K\"ahler and Integrability}
It is now clear that to accomodate the string geometry that follows from the metastring \cite{Freidel:2014qna,Freidel:2015pka} and to reunite the central ideas behind DFT and GG, we need to
include in our geometrical framework -- in addition to the $O(d,d)$ metric $\eta$ of DFT -- a two-form $\omega$. This two-form can be thought of as giving a preferred splitting of $T\PS$ and is a central object that allows one to pass from the local constructions of DFT to a globally well-defined geometry, even when $\eta$ is not a flat metric. The triple $(\PS,\eta,\omega)$ forms a {\it para-Hermitian} structure and the goal of this section is to review and discuss its relevant properties.

As we will see the simplest generalisation that allows for a curved $\eta$, is to demand that $\omega$ is closed and that its Lagrangian submanifolds are integrable. This corresponds to a para-K\"ahler geometry. In this case the manifold $\PS$ is locally split  $\PS\simeq M\times \tilde{M}$ and the leaves $M$ and $\tilde{M}$ are  affine manifolds \cite{Weinsteinsymp}. The curvature of $\eta$ shows up in the non-constancy of the pairing between $TM$ and $T\tilde{M}$. This is a natural generalisation of the DFT geometry setting.
If we want to go beyond this case, we have to relax the condition that $\omega$ is symplectic or that the distributions are integrable. The interplay between the almost symplectic geometry of $\omega$ and the  $O(d,d)$ geometry of $\eta$ is an important theme of our work.

We refer the reader to the review paper \cite{Cruceanu} for some standard results on para-Hermitian manifolds and to the article \cite{Etayo:2004lja} for the relationship with bi-Lagrangian structures. We also acknowledge that it was first noted by Vaisman \cite{Vaisman:2012px, Vaisman:2012ke} that a natural
framework for the geometrical understanding of DFT geometry is given by para-Hermitian
geometry. Our work and results can be viewed as a continuation of his work, also see \cite{inprogress}. 
Let's start by the definition of a para-Hermitian geometry:

\begin{Def}
An almost  para-Hermitian manifold $(\PS,\eta,K)$ is a manifold $\PS$ equipped with a neutral metric $\eta$ and an endomorphism $K \in \mathrm{End}(T\PS)$ such that 
\begin{equation}
K^2 =+\id, \qquad K^T \eta K =- \eta. 
\end{equation} 
\end{Def}
\noindent Since $K$ is a real (or para-complex) structure, it has eigenvalues $\pm 1$ and the tangent bundle of $\PS$ splits as $T\PS=L\oplus \Lt$, where $L$ and $\Lt$ are the $+1$ and $-1$ eigenbundles of $K$. 
The second condition implies that the distribution $L$ is maximally isotropic 
with respect to $\eta$. In other words 
\be
\eta(X,Y)=0,\qquad \forall\, X,Y\in L. 
\ee
Similarly, $\Lt$ is also maximally isotropic. Since $\eta$ is neutral, $L$ and $\Lt$ have the same rank $d=\frac{1}{2}\mathrm{Dim}\, \PS$.  The axioms further imply that
\be\label{omega}
\omega:= \eta K \quad\mathrm{or}\quad \omega(X,Y)= \eta(K(X),Y)
\ee
is an almost symplectic structure, i.e. a non-degenerate two-form. We can also define natural projections onto $L$ and $\Lt$ as 
\be
P := \frac12 (\id + K) \quad\mathrm{and}\quad \tilde{P}:=\frac12 (\id - K). 
\ee
Now the compatibility between $\eta$ and $\omega$ also imposes that $L$ and $ \Lt$ are {\it Lagrangian}\footnote{Even if $\omega$ is not closed we can still defined a Lagrangian subspace of $\omega$ as a space $L$ such that $L^\perp=L$.}  with respect to $\omega$.
We see that an almost para-Hermitian structure could have also been equivalently defined as an almost symplectic manifold $(\PS,\omega)$ with a compatible real structure $K$. When no confusion is possible, we will sometimes refer to $(\PS,\eta,K)$ as $(\PS,\eta,\omega)$, where $\omega=\eta K$.

More conceptually, an almost para-Hermitian structure on a manifold is a reduction of the structure group of the tangent bundle $T \PS$ from GL$(2d, \mathbb{R})$ to 
 \be 
 {\rm{O}}(d,d,\mathbb{R}) \cap {\rm{Sp}}(2d,\mathbb{R})= {\rm{GL}}(d,\mathbb{R}).
 \ee 
The compatibility between $\eta$ and $\omega$ is central for para-Hermitian geometry. In order to understand this geometry one needs to understand the properties of the connections  compatible with both structures. We review these properties here, details and proofs of the given statements can be found in Appendix \ref{sec:connections_apndx}. 
A general discussion and classification of compatible connections and classifications can be found in \cite{Ivanov:2003ze}, also see \cite{Berman:2013uda} for connections in the context of $O(d,d)$ geometry.

Given an almost para-Hermitian manifold, we can define compatible connections:
\begin{Def} A para-Hermitian connection on an almost  para-Hermitian manifold \\ $(\PS,\eta,K)$ is a connection $\n$ which preserves both $\eta$ and $\omega$, i.e. 
\be
\n\eta =\n\omega=0. 
\ee
Alternatively, a para-Hermitian connection preserves $\eta$ and $K$. 
The Levi-Civita connection $\lc$ is the unique connection which preserves $\eta$ and is torsionless. 
\end{Def}
%Given a para-Hermitian connection $\n$ we define its contorsion tensor $\hat\Omega$ to be the $(2,1)$ tensor given by 
%\begin{equation}
%  \eta(\n_XY,Z)=\eta(\lc_XY,Z)+\eta(X,\hat\Omega(Y,Z)).
%\end{equation}
%We also denote by $\Omega$ the corresponding $(3,0)$ tensor 
%$\Omega(X,Y,Z):= \eta(X,\hat\Omega(Y,Z))$.
%\end{Def}
Since we have a metric, the Levi-Civita connection is a natural connection to consider. However,
in general the Levi-Civita is not compatible with $\omega$ and therefore is not a para-Hermitian connection. Demanding that $\lc$ preserves $\omega$ imposes that  $\omega$ is closed. Therefore, a necessary (but not sufficient, see Corollary \ref{para-Kahler}) condition for the Levi-Civita connection $\lc$ to be para-Hermitian is that $\omega$ is symplectic.  This follows from the following lemma:
\begin{Lem}\label{firstlemma}
Let $(\PS,\eta,K)$ be an almost para-Hermitian manifold with compatible two-form $\omega$. The Levi-Civita derivative of $\omega$ satisfies the following properties
\begin{align}
\lc_X\omega(P(Y),\tilde{P}(Z)) &= 0, \label{eq:omegaPtP}\\
\lc_X\omega(Y,Z)&= \eta((\lc_XK) Y , Z), \label{omegaK}\\
\rd \omega(X,Y,Z) &= \sum_{(X,Y,Z)} \lc_X\omega(Y,Z) \label{rdom}
\end{align}
where $\sum_{(X,Y,Z)}$ is the sum over all cyclic permutations of $(X,Y,Z)$.
\end{Lem}
The first property emphasises that unlike $\omega$, $\lc\omega$ is of diagonal support with respect to the decomposition $T\PS= L\oplus \Lt$. The second property relates the derivatives of $\omega$ and $K$. The third one is standard, it simply relates the differential and covariant derivative when the connection is torsionless. This lemma is proven in appendix \ref{sec:connections_apndx}.

It is a well-known fact that any metrical connection $\n$ of the metric $\eta$ can be split in terms of the Levi-Civita connection $\lc$ and a so-called {\it contorsion tensor} $\hat{\Omega}$ which is a $(2,1)$ tensor. By definition it is given by
\begin{equation}
\eta(\n_XY,Z)=\eta(\lc_XY,Z)+\eta(X,\hat\Omega(Y,Z)).\label{eq:contorsion}
\end{equation}
We also denote by $\Omega$ the corresponding $(3,0)$ tensor, $\Omega(X,Y,Z):= \eta(X,\hat\Omega(Y,Z))$.
The relationship between a para-Hermitian connection and the Levi-Civita connection of $\eta$ can then be expressed purely in terms of this tensor:
\begin{Lem}\label{lemma:levicivita_cont}
Any para-Hermitian connection $\n$ has a skew contorsion tensor, ${\Omega}(X,Y,Z)=-{\Omega}(X,Z,Y)$, whose diagonal components are  determined by  the Levi-Civita derivative of $\omega$ as follows
\bea \label{omeq}
%-\lc_X \omega(Y,Z)&=& \Omega(X,Y,K(Z)) + \Omega(X,K(Y),Z).\cr
\lc_X \omega(P(Y),P(Z))&=& 2\Omega(X,P(Y),P(Z)), \cr
\lc_X \omega(\tilde{P}(Y),\tilde{P}(Z))&=& -2\Omega(X,\tilde{P}(Y),\tilde{P}(Z)).
\eea 
 \end{Lem}
These properties imply that only  the off-diagonal components of the contorsion tensor   are not determined by $\omega$.

\subsection{Integrability, closure  and a key example}

On the surface,  the notion of an almost para-Hermitian manifold is very similar to the notion of an almost Hermitian manifold $(\tilde{\PS}, \HH, I)$, where $\HH$ is an Euclidean metric and $I$ is an almost complex structure with compatibility $I^T\HH I=\HH$, $I^2=-1$ and with structure group $U(d)$. In the almost Hermitian case, there are two levels of of integrability one usually studies: if the complex structure is integrable, the manifold is said to be a Hermitian manifold; if the two-form $\omega$ is closed, the manifold is said to be almost K\"ahler. When both conditions are realised the manifold is said to be K\"ahler \cite{Moroianu}.

Therefore, the same characterisation of almost para-Hermitian manifolds can be implemented. Integrability of the real structure $K$ is encoded into the real analogue of the Nijenhuis tensor
\be \label{eq:nijenhuis}
4 N_K (X,Y) :=   [K(X),K(Y)] +[X,Y] -K([K(X),Y]+[X,K(Y)]).
\ee
When this tensor vanishes, the manifold is said to be para-Hermitian. Moreover, if the two-form $\omega$ is closed, the manifold is said to be almost para-K\"ahler. When both conditions are realised the manifold is said to be para-K\"ahler. These different situations have been studied extensively, see \cite{para-Kahler, Ivanov:2003ze} for definitions and examples. When the manifold is para-Hermitian it is locally split, i.e. it admits a set of local coordinates $(x,\tilde x) $ for which the product structure reads $K(\partial_x)=\pa_x$ and $K({\pa}_{\tilde{x}} ) = - \pa_{\tilde{x}}$.

Despite the similarities, the geometry of the para-Hermitian structure is fundamentally different from its complex analogue. There are two reasons behind this: in the complex case the two distributions  that diagonalise $I$ and decompose as $T\PS \otimes \mathbb{C}$ are conjugate to each other, so one is integrable if and only if the other one is. In the para-Hermitian case the integrability of the two distributions $L$ and $\tilde{L}$ are independent from one another, i.e. one Lagrangian distribution can be integrable while the other is not. This can be best seen by rewriting the Nijenhuis tensor \eqref{eq:nijenhuis} as a sum associated with each Lagrangian 
%\begin{equation}
%N_K (X,Y)=\PPt[\PP(X),\PP(Y)]+\PP[\PPt(X),\PPt(Y)].
%\end{equation}
%We therefore see that 
%$N_K(X,Y)$ splits to two parts, 
$N_K(X,Y)=N_{\PP}(X,Y)+N_{\PPt}(X,Y)$ where
\begin{equation}\label{eq:nijenhuis_split}
N_{\PP}(X,Y)\coloneqq \PPt[\PP(X),\PP(Y)],\quad N_{\PPt}(X,Y)\coloneqq \PP[\PPt(X),\PPt(Y)] .
\end{equation}
Now $N_{\PP}\in \Lt $ governs the integrability of $L$, while $N_{\PPt}\in L$ the integrability of $\Lt$ and we see that one can vanish even if the other one does not.
%Moreover, by the Frobenius theorem, the integrability of  $L$  implies the underlying manifold is foliated by leaves to which $L$ is tangent at every point \cite{Molino}. These leaves are, as we will see, equipped with a canonical flat connection: the Bott connection. This extra structure induced by the integrability of $L$ plays a key role in understanding the geometry of ``half-integrable'' para-Hermitian structures. 
In order to account for the possibility of ``half-integrability'', we introduce some new terminology.
\begin{Def}
An almost para-Hermitian manifold $(\PS,\eta,\omega)$ is said to be  $L$-para-Hermi-tian if $L$ is integrable. Also we say that an almost para-Hermitian manifold is almost $L$-para-K\"ahler if the conditions 
\be\label{LparaKahler}
 \rd \omega(\PP(X),\PP(Y),\PP(Z)) =0,\qquad \rd \omega(\PP(X),\PPt(Y),\PPt(Z)) =0,
\ee
are satisfied. If a manifold is $L$-para-Hermitian and $L$-para-K\"ahler then it is said to be almost  $L$-para-K\"ahler.
Similarly we can define the notion of $\tilde{L}$-para-Hermitian, almost $\Lt$-para-K\"ahler and  $L$-para-K\"ahler, by echanging the role of $L$ and $\Lt$.
\end{Def}
From this definition it is clear that if a manifold is both $L$ and $\Lt$-para-hermitian then it is simply para-hermitian. When in addition $L$ is integrable, the manifold is said to be $L$-para-K\"ahler (or $\Lt$-para-K\"ahler if $\Lt$ is integrable).
or $\tilde{L}$-para-Hermitian if $\tilde{L}$ is integrable. When both $L$ and $\tilde{L}$ are integrable the manifold is  para-Hermitian and similarly for the para-K\"ahler and almost para-K\"ahler cases. 
The fact that an $(L,\Lt)$-para-K\"ahler manifold is para-K\"ahler is needed in the next section where we study the relationship between the integrability conditions  and properties of $\rd \omega$. Before doing so, we present a key example.

\paragraph{Example:} A canonical example (see also \cite{Vaisman:2012px,vilcu2011hyperhermitian}) that illustrates the usefulness of the notion of ${\Lt}$-integrability is when $\PS = TM$ is the total space of the tangent bundle of a $d$-dimensional manifold $(M, g, D)$ equipped with a metric $g$ and a compatible connection $D$. We denote by $\pi: \PS \to M$ the canonical projection, by $\rd \pi: T\PS \to TM$ its differential and by $\tilde{L}=\mathrm{Ker}(\rd \pi)$ the vertical bundle. {In physical terms the vertical bundle is the space of velocities and it is canonically identified with the tangent bundle since  the vertical fiber is $\pi^{-1}(x)=T_xM$. }  The affine connection $D$ can be viewed as an Ehresmann connection, i.e. a decomposition $T\PS= L\oplus \tilde{L}$ of the tangent bundle in terms of a sum of a horizontal\footnote{ The affine connection $D$ associated with the decompostion $T\PS= L\oplus \tilde{L}$ is   defined by $ {P}( \rd Z(X))= D_X Z$ where $X,Z\in TM$.  Here  $Z\in TM $ is viewed as a map $Z:M\to \PS$ and $\rd Z:TM\to T\PS$ denotes its differential. $D$ can be viewed as a map $D: T\PS \to TM$ whose kernel defines the horizontal bundle.} bundle ${L}$ and a vertical bundle ${\Lt}$ (see \cite{ehresmann1948connexions,  Gudmundssson}).   
 
Given a point $x \in M$ and a point $p=(x,v)\in TM$ with $v\in T_xM$, we can define the notion of horizontal and vertical lift $\mathsf{h}:TM \to {L}$ and $\mathsf{v}: TM \to \tilde{L}$. These are defined as follows:
\begin{Def}\label{def:lifts}
The horizontal lift of $X\in T_xM$ to  $p=(x,v) \in TM$ is the unique vector $\mathsf{h}(X) \in {L}_p$ such that $\rd\pi(\mathsf{h}(X)) = X$. The vertical lift $\mathsf{v}(X)$
of $X\in T_xM$ to  $p=(x,v) \in TM$ is the unique vector $\mathsf{v}(X) \in \tilde{L}_p$ such that
$\mathsf{v}(X)[df]= X(f)$ for $f\in C^\infty(M)$. Here $\rd f$ is the linear function $\rd f (x,v)= v(f)$.
\end{Def}
We can then promote the  point-wise construction of the previous definition to any vector field $X\in \se(TM)$ to obtain its horizontal and vertical lifts $\mathsf{h}(X)\in \se(L),\ \mathsf{v}(X)\in \se(\Lt)$. These vector fields can then be equipped with a $\PS$-Lie bracket (i.e. a Lie bracket on $T\PS=TTM$), satisfying the relations \cite{Sasaki}
\bea
[\mathsf{v}(X), \mathsf{v}(Y)]_{\PS}=0, \quad [\mathsf{h}(X), \mathsf{v}(Y)]_{\PS}= \mathsf{v}(D_XY),  \cr
[\mathsf{h}(X),\mathsf{h}(Y)]_{\PS} = \mathsf{h}([X,Y]) -  \mathsf{v}[R_D(X,Y)v],
\eea
where $R_D(X,Y)v= [D_X,D_Y]v-D_{[X,Y]}v$ is the curvature of $D$. We see that the set of vertical vectors $\tilde{L}$ forms an ideal of $T\PS$.
We also see that the set of horizontal vectors is integrable if and only if the curvature of $D$ vanishes. 

We will now also show that the lifting maps $\mathsf{h},\mathsf{v}$ given in Definition \ref{def:lifts} are well defined. For this purpose, we choose local coordinates $x^a$ on $M$ and  corresponding local coordinates $(x,v)$ on $TM$, where $v=v^a \pa_a$ and we denote $\pa_a :=\frac{\pa}{\pa x^a}$ and $\tilde{\pa}_a :=\frac{\pa}{\pa v^a}$. The projection map in these coordinates is given by $\pi(x,v)=x$.  An element $\mathbb{X}$ of $\Gamma(T\PS)$ can be decomposed as $\mathbb{X}(x,v)= X^a(x,v)\pa_a + V^b(x,v)\tilde{\pa}_b$. The vertical and horizontal maps associated with $X= X^a(x)\pa_a\in \Gamma(TM) $ are then explicitly given by 
\be
\mathsf{v}(X)(x,v) = X^a {\tilde\pa_a}, \quad
\mathsf{h}(X)(x,v) =  X^a (\pa_a - \Gamma_{a\bm{v}}^c  \tilde{\pa_c})
%=X^a (\pa_a - v(D_X v) \pa_b v^b  \frac{\pa}{\pa v^c})
\ee
where $\Gamma_{ab}^c$ are the connection symbols: $D_{\pa_a} \pa_b = \Gamma_{ab}^c \pa_c$, and we denoted $\Gamma_{a\bm{v}}^c:= \Gamma_{ab}^c v^b$ for simplicity. In order to check that $\mathsf{h}(\pa_a)$ is horizontal we recall that a vector $p(t)=(x(t),v(t))\in TM$ is transported parallelly along a path $x(t)\in M$ 
if $D_{\dot{x}} v=0$. This translates into the condition $\dot{v}^c=- \dot{x}^a \Gamma_{a \bm v}^c$, where the dot expresses the time derivative.  Therefore, the tangent components  to a parallel transport are in the image of $\mathsf{h}$: $\dot{p}(t)=(\dot{x},\dot{v}) = \mathsf{h}(\dot{x}^a\pa_a)$. The result then follows by direct computation.

We have seen that to curve the geometry of $\PS$, we need to consider half para-Hermitian manifolds. The above example is an instance of such a geometry:
\begin{Lem}
$(TM, g,D)$ is canonically equipped  with an $\tilde{L}$-para-Hermitian structure \\ $(TM,\eta,K)$
given by 
\be
\eta_p(\mathsf{v}(\pa_a), \mathsf{h}(\pa_b)) = g_{ab}(x), \quad K(\mathsf{h}(\pa_a)) = +\mathsf{h}(\pa_a),\quad
K(\mathsf{v}(\pa_a)) = -\mathsf{v}(\pa_a).
\ee
$(TM,\eta,K)$ is para-Hermitian if and only if the curvature of $D$ vanishes. It is $\Lt$-para-K\"ahler if  $D$ is torsionless. It is almost para-K\"ahler if and only if $D$ is torsionless.
\end{Lem} 
\begin{proof} We have already seen that $\tilde{L}$ is integrable and that $L$ is integrable if and only if the curvature of $D$ vanishes which establishes the first claim. The pair dual to the basis $(\mathsf{h}(\pa_a), \mathsf{v}(\pa_a))\in TM$ is given by 
$(\rd x^a, D v^a)\in T^*\PS$, where 
\be
Dv^a :=\rd v^a + \Gamma_{b\bm{v}}^a \rd x^b  .
\ee
This can be checked by direct evaluation, for instance 
\bea
\mathsf{h}(\pa_a) [D v^b] &=& (\pa_a -\Gamma_{a\bm{v}}^c  \tilde{\pa}_c)[\rd v^b + \Gamma_{d \bm{v}}^b \rd x^d ]\cr
&=&\Gamma_{d\bm{v}}^b \pa_a[\rd x^d] -\Gamma_{a\bm{v}}^c \tilde{\pa}_c[\rd v^b] =0.
\eea
The para-Hermitian structure on $TM$ can also be described as the almost bi-Lagrangian structure $(\omega,K)$ where $ K(\rd x^a)=\rd x^a$, $K(Dv^a)=-Dv^a$ and the almost symplectic form is  given by 
\be
%\eta =  g_{ab}(x) (\rd x^a \otimes D v^b + Dv^b\otimes \rd x^a), \qquad 
\omega =   g_{ab}(x) \rd x^a \wedge  D v^b.
\ee
%The differential on $\PS$ can then be written as 
%\be \rd_\PS= \rd x^a \pa_a + \rd v^a \tilde{\pa}_a
%= \rd_\mathsf{h} + \rd_\mathsf{v}
%\ee where $\rd_\mathsf{h}= \rd x^a \mathsf{h}(\pa_a)$ is the horizontal differential and $\rd_\mathsf{v} = D v^a \mathsf{v}(\pa_a)$ is the vertical differential. We can evaluate that $\rd_\PS \rd x^a =0$ and $\rd_\PS D v^a = - \Gamma_{bc}^a \rd x^b \wedge D v^c +\frac12  \rd x^b\wedge  \rd x^c R_{bc \bm{v}}{}^a 
%$.  
%Using the metricity condition which implies that $\pa_a g_{bc}=\Gamma_{ab}^d g_{dc}+ \Gamma_{ac}^d g_{db}$ and $ R_{abcd}= -R_{abdc}$ and introducing the torsion $T_{ab}^c =\Gamma_{[ab]}^c$ 
It is straightforward to compute the differential of $\omega$
\be
\rd \omega = g_{cd}T_{ab}^d \rd x^a \wedge \rd x^b \wedge D v^c 
+ \frac12 R_{abc\bm{v}} \rd x^a \wedge \rd x^b \wedge \rd x^c 
\ee
where we have introduced the torsion $T_{ab}^c =\Gamma_{[ab]}^c$ and the curvature tensor  $ R_{abc\bm{v}}= R_{abcd} v^d$.
From this expression we see that $\rd\omega(\tilde{P}(X),\tilde{P}(Y),\tilde{P}(Z))=0$, and $\rd\omega(\tilde{P}(X),{P}(Y),{P}(Z))$ vanishes if  the torsion vanishes.
Since $\tilde{L}$ is integrable this establishes that $TM$ is $\Lt$-para-K\"ahler if $D$ is torsionless. It is also clear that the condition $\rd \omega=0$ implies that the torsion vanishes. On the other hand, if the torsion vanishes the second term in this expression also vanishes by the Bianchi identity. This concludes that the structure is almost para-K\"ahler if and only if the torsion vanishes.
\end{proof}

\subsection{$L$-para-Hermitian manifolds }
In the almost para-Hermitian geometry, the two distributions $L$ and $\tilde{L}$ play an equivalent role. For the purpose of our presentation it will be useful to  think from now on of $L$ as being ``tangent to the space-time base'' and $\tilde{L}$ as being ``tangent to the momentum fibre'' as in the previous example.

The integrability of $L$, which is insured if $N_{\PP}=0$ (see \eqref{eq:nijenhuis_split}), means that there exists a foliation by $d$-dimensional leaves ${\cal F}$ such that $L$ is tangent to the leaves, in other words $T{\cal F} =L$.
There is a profound interplay between the integrability of $L$ and properties of $\omega$. For instance, if $L$ is integrable then the pull back of $\rd \omega$ on ${\cal F}$ must vanish. These  relationships are investigated in the following lemmas. In the following, one uses the definition $N_{\PP}(X,Y,Z):= \eta(N_{\PP}(X,Y),Z)$.
\begin{Lem}
The cyclic permutation of the projected  Nijenhuis tensor is given by
\be
\sum_{(X,Y,Z)} N_{\PP}(X,Y,Z) 
= \rd \omega(P(X),P(Y),P(Z)).
\ee
Given a para-Hermitian connection $\n$ with torsion $T(X,Y,Z)=\eta(\hat{T}(X, Y) ,Z)$ we have the relationship
\be\label{NT}
N_{\PP}(X,Y,Z)= -T(P(X),P(Y),{P}(Z)).
\ee
\end{Lem}
\begin{proof}
From the definition of $N$ we have that
\bea
N_{\PP}(X,Y,Z)&=& \eta(\n_{P(X)}P(Y)-\n_{P(Y)}P(X)- \hat{T}(P(X),P(Y)),P(Z))\label{Nnlc}\\
&=& \eta(P(\n_{P(X)}Y)-P(\n_{P(Y)}X),P(Z))-  {T}(P(X),P(Y),P(Z))\nonumber
\eea
The second equality follows from the fact that the connection is para-Hermitian, and the first two terms in the last equality vanish because $L$ is Lagrangian. This establishes (\ref{NT}) and the first part of the lemma follows from (\ref{domega}) in the appendix, which establishes that 
\be
\rd\omega(X,Y,Z) = -\sum_{(X,Y,Z)} {T}(X, Y, K(Z)).
\ee
\end{proof}
\noindent The relationship between $N_\pr$ and $\lc\omega$ is encoded in the following formula
\begin{equation}
\begin{aligned}
2 N_\pr (X,Y,Z) 
&= \lc_{P(X)}\omega(P(Y),P(Z))- \lc_{P(Y)}\omega(P(X),P(Z)) \\
&= \rd \omega (P(X),P(Y), P(Z))- \lc_{P(Z)} \omega(P(X),P(Y)) .
\end{aligned}
\label{Nddw}
\end{equation}
This can be seen starting from the definition  (\ref{Nnlc}) and using $\lc$ instead. Since the torsion of the Levi-Civita connection vanishes and $2\eta(\lc_{P(X)}P(Y),P(Z))
=\eta( (\lc_{P(X)}K)Y,P(Z) ) $, we obtain the first equality from (\ref{omegaK}).
The second equality follows directly from the definition of the exterior derivative.
 This  leads to the following corollary: 
\begin{Cor}\label{wclosed}
The para-Hermitian structure $(\eta,\omega)$  is $L$-para-K\"ahler if and only if  the Levi-Civita connection of $\eta$ is symplectic along $L$ : $\lc_{P(X)} \omega=0$.
\end{Cor}
\begin{proof}
If the Levi-Civita connection of $\eta$ is symplectic along $L$ then one sees directly that the conditions (\ref{LparaKahler}) are satisfied thanks to Lemma \ref{firstlemma}, and that $N_P$ vanishes due to \eqref{Nddw}, which means that $(\eta,\omega)$ is $L$-para-K\"ahler. On the other hand being $L$-para-K\"ahler implies that both $N_\pr$ and  the pull-back of $\rd \omega$ on ${\cal F}$ vanish. From  (\ref{Nddw}) it follows that $\lc_{P(Z)} \omega(P(X),P(Y))=0$. Now, Lemma \ref{firstlemma} implies that the components $\lc_{P(Z)} \omega(P(X),\tilde{P}(Y))$ also vanish and that $\rd\omega(P(X),\tilde{P}(Y),\tilde{P}(Z)) = \lc_{P(X)}\omega(\tilde{P}(Y),\tilde{P}(Z))$. This last component also vanishes by the second condition of (\ref{LparaKahler}). This establishes the corollary. 
\end{proof}
If $(\PS,\eta,\omega)$ is para-K\"ahler then it is $(L,\Lt)$-para-K\"ahler. Using the previous lemma twice we can then conclude that 
\begin{Cor}\label{para-Kahler}
The Levi-Civita connection is a para-Hermitian connection if and only if 
$ (\PS,\eta,\omega)$ is para-K\"ahler.
\end{Cor}
%{\color{blue} Beware that these lemma have been changed so  lets check them carefully.}
This corollary is of central importance to us. The para-K\"ahler case is the closest generalisation of DFT geometry. The previous results shows that in this case the Levi-Civita connection is the natural connection to use since it preserves $(\eta,\omega)$.  

\subsection{The para-K\"ahler Case}
The para-K\"ahler manifolds, where both $L$ and $\Lt$ are integrable and $\rd \omega=0$, are an important subset of examples straightforwardly generalising the flat DFT geometry. 
For this section, we assume that $(\PS,\eta,\omega)$ is a para-K\"ahler manifold. As Corollary \ref{para-Kahler} establishes, this means that the Levi-Civita connection is the natural connection to use since it preserves $(\eta,\omega)$ while $\omega$ is symplectic.  We now explore further properties of $\lc$ on a K\"ahler manifold.
 
First, the condition of being para-Hermitian implies that there exists a double foliation $({\cal F},{\cal \tilde{F}})$ with $T{\cal F}= L$ and $T{\tilde{\cal F}}=\Lt $. In practice, this means that there exist local coordinates $(x^a,\tilde{x}_a)$ such that $K(\pa_a) = + \pa_a$ and $K(\pa_a)= - \tilde{\pa}^a$,
where we have denoted $\pa_a = \frac{\pa}{\pa{x^a}}$ and $\tilde{\pa}^a = \frac{\pa}{\pa{\tilde{x}_a}}$. In these coordinates, the leaves of $ \cal{F}$ are given by $\tilde{x}^a=\mathrm{const.}$ and similarly for $\cal{\tilde{F}}$ and $x^a$. This in turn implies that $\eta$ and $\omega$ take off-diagonal form and can be expressed in terms of functions $\eta_{a}{}^b(x,\tilde{x})$ as 
\be 
\eta=  \eta_{a}{}^b(x,\tilde{x})(\rd x^a\otimes \rd \tilde{x}_b  +  \rd \tilde{x}_b\otimes\rd x^a), \qquad
\omega=  \eta_{a}{}^b(x,\tilde{x})\rd x^a\wedge  \rd \tilde{x}_b.
\ee
It is convenient to define the partial differentials $\mathring{D} := \rd x^a\lc_{\pa_a} $ and $\mathring{\tilde{D}} :=\rd \tilde{x}^a \lc_{\tilde\pa^a}$.
The importance of these differentials lies in the following lemma:
\begin{Lem}\label{FlatK} If 
$(\PS,\eta,\omega)$ is para-K\"ahler then the partial differential  $\mathring{D}$ and $\mathring{\tilde{D}}$ are flat. In other words, they satisfy
\be
D^2=0,\qquad \mathring{\tilde{D}}^2=0. 
\ee
\end{Lem}
\begin{proof}  
First, since the Levi-Civita connection is torsionless, it means that the condition $\rd \omega =0$  is equivalent to the conditions $\mathring{D}\omega= \mathring{\tilde{D}}\omega=0$. If one introduces the one-forms $\eta^b :=  \rd x^a\eta_a{}^b $ and $\tilde{\eta}_a :=\eta_a{}^b \rd\tilde{x}_b$, 
%which enters in the definition of $\omega$ as $\omega = \eta_a \wedge \rd \tilde{x}^a$ and $\omega= \rd x^a \wedge \tilde{\eta}_a$.  
one can easily see that the closure condition implies 
\be\label{DDeta}
\mathring{D}\eta_a=  \mathring{\tilde{D}}\tilde{\eta}^a =0.
\ee
Now the condition that $\eta_a{}^b$ is invertible means that $\eta^a$ form a basis of one-forms on $\cal F$. The previous condition implies that this basis is parallel with respect to $\mathring{D}$. This can be true if and only if the curvature of $\mathring{D}$ along $L$ vanishes, that is if $\mathring{D}^2=0$. The same conclusion is reached for $\mathring{\tilde{D}}$. This proves the lemma. 
\end{proof}

This lemma expresses that in the para-K\"ahler case even if $\eta$ is not flat the leaf geometry, i.e. the geometry induced on the leaf, is flat. This is related to the well-known fact in symplectic geometry that the leaves of a symplectic foliation admits a canonical flat connection called the Bott connection \cite{
%Weinsteinsymp,FoliationConnection, 
Bott,Weinstein,forger2013lagrangian}. This result of central importance means that leaves of a symplectic foliation are affine manifolds. The previous lemma is a more direct derivation of this fact in the para-K\"ahler case.

Let's finally remark that the conditions (\ref{DDeta}) mean that the forms $\eta^a$ and $\tilde{\eta}_a$ are closed. Using the Poincar\'e lemma twice, we conclude that the tensor $\eta_a{}^b$ can locally be written in terms of a potential $\phi$
\be
\eta_{a}{}^b(x,\tilde{x}) =\pa_a \tilde{\pa}^b \phi(x,\tilde{x})
\ee
where $\phi$ is the para-K\"ahler analog of the K\"ahler potential \cite{Moroianu}.

\section{Connections on Para-Hermitian Manifolds}

In the previous section we have seen the introduction of the almost symplectic form $\omega$ which together with the metric $\eta$ gives rise to para-Hermitian geometry. We will now explore canonical constructions related to connections in this setting.

As we have seen, given an almost para-Hermitian manifold $(\PS,\eta, K)$, we can construct two projections, $\PP$ and $\PPt$ in $\mathrm{End}(T\PS)$, which satisfy $\PP+\PPt=\id$ and $\PP-\PPt=K$. Their respective images are $L$ and $\tilde{L}$.
%\begin{equation}
%\PP = \frac12(\id+K) \qquad\mathrm{and}\qquad \PPt = \frac12(\id-K) \, .
%\label{eq:projectors}
%\end{equation}
The compatibility condition $K^T\eta K=-\eta$ and the isotropy of $L$ and $\Lt$ translate into the following property
\begin{equation}
\eta(\PP(X),Y) = \eta(X,\PPt(Y)) = \eta(\PP(X),\PPt(Y)),
\label{eq:etaPPtilde}
\end{equation}
which will be needed frequently.

\subsection{Projected and Partial Connections}
\label{sec:genLie}

In order to define a generalisation of the Lie derivative and the corresponding bracket for the curved case, we  need the following objects.
\begin{Def} Given a metric $\eta$, a compatible connection $\nabla$ and a projector $\PP:T\PS \to T\PS$, we denote $L = \mathrm{Im}\,\PP$ and $\tilde{L} = {\mathrm Ker}\, \PP = \mathrm{Im}\,\PPt$ with $\PPt:= \id-P$ the complimentary projector. The {\bf projected derivative} $\Dbr_X: \Gamma(T\PS) \to \Gamma(T\PS)$ is defined by 
\begin{equation}
%\begin{aligned}
%  \Dbr : \Gamma(T\PS)\, \times & \, \Gamma(T\PS) &\rightarrow & \quad \Gamma(T\PS) \\
%  (X\, , &\, Y) &\mapsto & \quad 
  \Dbr_XY := \nabla_{P(X)} Y. 
%\end{aligned}
\end{equation}
Given a function $f\in C^\infty(\PS)$, the  projected differential $\Db f \in \Gamma(T\PS)$ is a vector field such that 
\be
\eta(X, \Db f) :=  \PP(X)[f],
\ee
where $X[f]$ denotes the action of the vector field $X$ on a function $f$.
\label{def:partialD}
\end{Def}
\noindent 
Note that $\Dbr_X$ satisfies the Leibniz rule
\be \label{Leibniz}
D_X(f Y) = P(X)[f] \, Y + f(D_X Y).
\ee
Viewing $\PP$ as the surjective map $T\PS\rightarrow L$ (where of course $L\hookrightarrow T\PS$), then $\Dbr_X$ defines a \emph{partial connection} along $L$ \cite{FoliationConnection,forger2013lagrangian}:
\begin{equation}
\begin{aligned}
  \n : \Gamma(L)\, \times & \, \Gamma(T\PS) &\rightarrow & \quad \Gamma(T\PS) \\
  (x\, , &\, Y) &\mapsto & \quad \n_xY := \Dbr_X Y. 
\end{aligned}
\end{equation}
for some $X\in\Gamma(T\PS)$ such that $P(X)=x\in\Gamma(L)$. Moreover, $\Db f$ is an element of $\tilde{L}$, i.e $\PP(\Db f)=0$. This can be seen by using \eqref{eq:etaPPtilde} 
\begin{equation}
  \eta(P(\Db f),X) = \eta(\Db f,\PPt(X)) = \PP(\PPt(X))[f] = 0 \, .
\end{equation}

Using this projected connection, we can define the analog of curvature, torsion and the Nijenhuis tensor (cf. \eqref{eq:nijenhuis_split} above):
\begin{Def}
Given a metric $\eta$, a compatible connection $\nabla$ and a projector $P$, the projected Riemann tensor $\hat{R}_\pr$, the projected torsion $\tau_\pr$ and the projected Nijenhuis tensor $N_\pr$ are elements of $T\PS$ given by 
\begin{align}
\hat{R}_\pr(X,Y)Z  &:= [\Dbr_X,\Dbr_Y] Z - \Dbr_{[\PP(X),\PP(Y)]}Z \\
\tau_\pr(X,Y) &:= \PP([\Dbr_XY-\Dbr_YX]-[\PP(X),\PP(Y)]),\\
N_\pr(X,Y)&:= \tilde{\PP}( [\PP(X),\PP(Y)]).
\end{align}
We also define ${R}_\pr(X,Y,Z,W):=\eta (\hat{R}_\pr(X,Y)Z, W)$ and we introduce the twist of the projected derivative
\begin{equation}
\theta_\Dbr(X,Y) := \eta(X,\Db Y),
\end{equation}
which satisfies $\eta(Z,\theta_\Dbr(X,Y))= \eta(X,\Dbr_Z Y)$. 
\label{def:projectedobjects}
\end{Def}

Note that  $L$, the image of $\PP$, forms an integrable distribution if and only if the projected Nijenhuis tensor $N_\pr(X,Y)$ vanishes. Since $\Dbr$ can be viewed as a partial connection $\Dbr: \Gamma(T\PS) \rightarrow \Gamma(L^*)\times \Gamma(T\PS)$, we can define the curvature of $\Dbr$ as $R_\pr = \Dbr^2: \Gamma(T\PS)\rightarrow \Gamma(\wedge^2 L^* \otimes \mathrm{End}\,T\PS)$ (where $L^*$ is the dual vector bundle of $L$). This curvature $R_\pr$ is related to $R$, the usual curvature of $\n$, via
\begin{equation}
\hat{R}_\pr(X,Y) = \hat{R}(\PP(X),\PP(Y)) \, .
\label{eq:curvature}
\end{equation}
The projected curvature $\hat{R}_\pr$ is a proper $(3,1)$ tensor as can be seen from
\begin{equation}
\begin{aligned}
\hat{R}_\pr(fX,Y)Z &= \hat{R}_\pr(X,fY)Z = f\hat{R}_\pr(X,Y)Z \\
\hat{R}_\pr(X,Y)(fZ) &= f\hat{R}_\pr(X,Y)Z + N_\pr(X,Y)[f] Z = f\hat{R}_\pr(X,Y)Z
\end{aligned}
\end{equation}
where for the last equality to hold we need integrability of $L$, i.e. $N_\pr(X,Y)=0$.

The projected torsion $\tau_\pr$ is a $(2,1)$ tensor valued in $L$ and it is easy to check that $\tau_\pr(fX,Y)=\tau_\pr(X,fY)=f\tau_\pr(X,Y)$. It can be expressed in terms of the ordinary torsion $T(X,Y)=\n_XY-\n_YX - [X,Y]$ as 
\be
\tau_\pr(X,Y)= T(\PP(X),\PP(Y)) - \frac12\big( (D_XK)Y - (D_YK)X\big) + N_\pr(X,Y).
\label{eq:torsionrelation}
\ee
Like $\Db f$, the projected twist vector $\theta_\Dbr(X,Y)$ is valued in $\tilde{L}$, i.e. $\PP(\theta_\Dbr)=0$
\begin{equation}
  \eta(Z,\PP(\theta_\Dbr(X,Y))) = \eta(\PPt(Z),\theta_\Dbr(X,Y)) = \eta(X,\Dbr_{\PPt(Z)}Y) = 0 \, .
\end{equation}
It is tensorial in $X$ while $\theta_\Dbr (X,fY) =  \eta(X,Y)\Db f + f \theta_\Dbr(X,Y)$.

The definition given for the projected torsion and curvature is still meaningful even if $L$  is not integrable, i.e. when $N_\pr$ does not vanish. Then $\hat{R}_\pr$ is related to the $(3,1)$ Riemann tensor $\hat{R}$ of $\n$ via
\be
 \hat{R}_\pr(X,Y)Z = \hat{R}(\PP(X),\PP(Y))Z +  \n_{N_\pr(X,Y)}Z 
\ee
which in general is of course not tensorial in $Z$.

\subsection{Generalized Lie Derivative and Bracket}
The main ingredient entering the construction of both generalised geometry and double field theory is the Dorfman derivative that generalises the notion of the Lie derivative. The purpose of our work is to generalise this notion to the case where $\eta$ is not necessarily flat. We can now use any metrical connection of $\eta$ to form a generalisation of the Lie derivative \cite{
%vaisman2005tangent,Grana:2008yw,Hull:2009zb,
Vaisman:2012ke}.

\begin{Def} Given a metric $\eta$, a compatible connection $\nabla$ and a projector $\PP$, the associated {\bf generalised Lie derivative} (Dorfman derivative) $\Lb^\Dbr : \Gamma(T\PS) \times \Gamma(T\PS) \to \Gamma(T\PS)$  is given by  
\begin{equation}
\begin{aligned}
  \Lb^\Dbr_XY &:=\Dbr_{X}Y-\Dbr_{Y}X+ \theta_\Dbr(Y, X),
  %\\
 % L_XY &= \n_{\PP(X)}Y-\n_{\PP(Y)}X+Y \Dbr X, \qquad \Dbr_X= \nabla_{\PP(X)}
\end{aligned}
\end{equation}
where the twist $\theta_\Dbr(Y,X):= \eta(Y, \Db X)$ is such that $\eta(Z, \theta(Y,X)) = \eta(Y, \Dbr_Z X)$.
\label{def:genLie}
\end{Def}

This generalised Lie derivative satisfies three key properties: the Leibniz rule, compatibility with $\eta$ and a normalisation condition which are respectively given by
\begin{equation}
\begin{aligned}
\Lb^\Dbr_X (f Y) &=  f (\Lb^\Dbr_XY)  + \PP({X})[f]\,Y, \\
\PP(X) [\eta(Y,Z)] &= \eta( \Lb^\Dbr_X Y, Z) + \eta(Y, \Lb^\Dbr_XZ), \\
\eta( Y, \Lb^\Dbr_XX)  &= \frac12 \PP(Y)[\eta(X,X)].
\end{aligned}
\label{eq:genLieproperties}
\end{equation}
The first property follows directly from the fact that only the first term in the definition of $\Lb^\Dbr$ contains a derivative of $Y$. The second property -- which implies the preservation of $\eta$  for any vector field -- distinguishes the generalised Lie derivative from the usual Lie derivative for which such a condition implies that $X$ is a Killing vector field. It can be simply proven by summing the following two terms
\be
\begin{aligned}
 \eta( \Lb_Z^\Dbr X, Y) &=& \eta(\Dbr_Z X, Y)-\eta(\Dbr_X Z,Y) + \eta(X, \Dbr_Y Z),\cr
 \eta(X,\Lb^\Dbr_Z Y) &=& \eta(X, \Dbr_Z Y) -\eta(X, \Dbr_Y Z) + \eta(Y, \Dbr_X Z),
\end{aligned}
\ee
and using  the compatibility of $\n$ with $\eta$. Moreover, unlike the usual Lie derivative, a generalised Lie derivative of a vector field along itself does not necessarily vanish due to the third property. 

% On the other hand, the generalised Lie derivative along a vector of the form $\Db f$ always vanishes if the projected torsion of the connection vanishes. This can be seen from
% \begin{align}
% \eta( \Lb_{\Db f}^\Dbr X, Y) &=  \eta(\Dbr_{\Db f} X, Y) - \eta(\Dbr_X \Db f ,Y) 
% 		+ \eta(X, \Dbr_Y \Db f) \notag\\
%  &= \eta(\Db f, \theta_\Dbr(Y,X)) - P(X)[ \eta( \Db f ,Y)]  + \eta( \Db f ,\Dbr_X Y)  \notag\\
%  &\qquad + P(Y)[\eta(X, \Db f)] -  \eta(\Dbr_Y X,  \Db f)  \notag\\
%  &= \eta( \Db f , \Dbr_X Y- \Dbr_Y X - [P(X),P(Y)] + \theta_\Dbr(Y,X))  \notag\\
%  &= \eta( \Db f ,\tau_\pr(X,Y)) + \eta(\Db f, \theta_\Dbr(Y,X))
% \end{align}
% where we have used \eqref{eq:etaPPtilde} and the fact that both $\Db f$ and $\theta_\Dbr(Y,X)$ are valued in $\tilde{L}$, hence $\PP(\Db f)=\PP(\theta_\Dbr)=0$. Then the last term vanishes identically, the other term disappears for torsionless connections.

The generalised Lie derivative $\Lb^\Dbr$ does not satisfy the Jacobi identity in general. In order to measure its failure of the Jacobi identity to be satisfied, we introduce the so-called Jacobiator.
\begin{Def}
Given a generalized Lie derivative $\Lb^\Dbr$ as in Definition \ref{def:genLie}, the {\bf Jacobiator} is defined as
\begin{equation}
J^\Dbr(X,Y,Z,W) := \eta([\Lb^\Dbr_X,\Lb^\Dbr_Y]Z - \Lb^\Dbr_{\Lb^\Dbr_XY}Z,W) \, .
\end{equation}
The conditions on $\Lb^\Dbr$ under which this object is tensorial are given in Appendix \ref{sec:jacobiator}. If they are satisfied, $J^\Dbr=0$ states that the Jacobi identity for $\Lb^\Dbr$ holds.
\label{def:jacobiator}
\end{Def}
We will take a detailed look at the Jacobiator for para-K\"ahler and para-Hermitian geometries in the next subsection.

\begin{Rem}
In generalised geometry and double field theory, the Dorfman derivative satisfies the Jacobi identity (under certain constraints). Unlike the ordinary Lie derivative, these generalised derivatives $\Lb^\Dbr_XY$ are not skew-symmetric under the interchange of $X$ and $Y$. Taking the skew-symmetric combination $\bl X,Y\br^\Dbr := \frac12(\Lb^\Dbr_XY-\Lb^\Dbr_YX)$ gives the Courant bracket \cite{courant1990dirac}, which is obviously skew but does not satisfy Jacobi identity. We will mainly work with the generalised derivative instead of the bracket as they are equivalent and related by
\begin{equation}
\Lb^\Dbr_XY = \bl X,Y\br^\Dbr + \frac12\Db\eta(X,Y) \, .
\end{equation}
\end{Rem}

\subsection{Para-K\"ahler and $L$-Para-Hermitian}
\label{sec:para}
Now that we have properly introduced and defined the generalisations of various concepts and quantities, we would like to use them to find connections $\n$ and corresponding (projected) generalized Lie derivatives $\Lb^\Dbr_X$ (with $\Dbr_{X}=\n_{\PP(X)}$) for the para-K\"ahler and para-Hermitian geometries. The Jacobiator $J^\Dbr$ from Definition \ref{def:jacobiator} is evaluated explicitly in Appendix \ref{sec:jacobiator}, we find
\begin{equation}
\begin{aligned}
J^\Dbr(X,Y,Z,W) &=  R_\pr(X,Y,Z,W) + R_\pr(Y,Z,X,W)  + R_\pr(Z,X,Y,W) \\
	&- R_\pr(W,Z,X,Y) - R_\pr(W,X,Y,Z) - R_\pr(W,Y,Z,X) \\
	&- \eta(W,\n_{\tau_\pr(X,Y)}Z) - \eta(W,\n_{\tau_\pr(Y,Z)}X)  
			- \eta(W,\n_{\tau_\pr(Z,X)}Y) \\
	&- \eta(Z,\n_{\tau_\pr(X,W)}Y) - \eta(X,\n_{\tau_\pr(Y,W)}Z) 
			+ \eta(Y,\n_{\tau_\pr(W,Z)} X) \\
	&+ \eta(\theta_\Dbr(Z,X),\theta_\Dbr(W,Y)) - \eta(\theta_\Dbr(Z,Y),\theta_\Dbr(W,X)) \\
	&- \eta(\theta_\Dbr(Y,X),\theta_\Dbr(W,Z)) \, 
\end{aligned}
\label{eq:jacobiator}
\end{equation}
using $R_\pr$, $\tau_\pr$ and $\theta_\Dbr$ from Definition \ref{def:projectedobjects}. Comparing this to the Jacobiator of DFT \eqref{eq:jacobiatorDFT}, we can see extra terms due to curvature and torsion of the connection $\n$. We recall that because $L$ and $\Lt$ are Lagrangian, $\theta_\Dbr \in \se(\Lt)$ and so the last two lines of \eqref{eq:jacobiator} vanish identically. We now analyse the possibility to construct a connection with vanishing Jacobiator in two different cases that are deeper generalisations of double field theory.

The simplest case is the case where we assume that $L$ and $\tilde{L} $ are integrable Lagrangian distributions and we also assume that $\omega$ is closed. In this case $(\PS,\eta,K)$ is said to be {\bf para-K\"ahler}. Because of integrability of $K$, the Nijenhuis tensor vanishes. Furthermore, the Levi-Civita connection $\lc$ of the metric $\eta$, preserves $\eta, K$ and $\omega$, as we have seen earlier in Corollary \ref{wclosed}. Therefore, its projected torsion $\mathring{\tau}_\pr$ vanishes (as can be seen from \eqref{eq:torsionrelation}). 
We already know from Lemma \ref{FlatK} that the projected curvature tensor $\mathring{R}_\pr$ vanishes. This can also be seen from an explicit calculation; since $\lc$ preserves $K$, we have (using \eqref{eq:etaPPtilde})
\begin{equation}
\mathring{R}(Z,W,P(X),P(Y))=\eta(\hat{\mathring{R}}(Z,W)\PP(X),\PP(Y))=\eta(\PP(\hat{\mathring{R}}(Z,W)X),\PP(Y))= 0.
\end{equation} 
Therefore, the symmetry of the Riemann tensor, $\mathring{R}(X,Y,Z,W)=\mathring{R}(Z,W,X,Y)$, and the identity (\ref{eq:curvature}), imply that 
\be
\mathring{R}_\pr(X,Y,Z,W) = 0.
\ee 
This shows that if $(\PS,\eta,\omega)$ is a para-K\"ahler manifold, then $\Lb^{\mathring{\Dbr}}_XY$ with $\mathring{\Dbr}_X=\lc_{P(X)}$ is a generalised Lie derivative that satisfies the Jacobi identity for any $X,Y \in \Gamma(T\PS)$.

We now relax the condition that $\omega$ is closed and  we also relax the integrability conditions on $\tilde{L}$. 
As we will show, in order to construct a Lie derivative that satisfies the Jacobi identity, we only need $L$ to be integrable. In this case, $(\PS ,\eta, K)$ is {\bf $L$-para-Hermitian}, and because $\lc \omega$ no longer vanishes, the \gld associated to the Levi-Civita connection no longer satisfies the Jacobi identity as in the para-K\"ahler case. We therefore have to find a different connection with this property:
\begin{Def}\label{def:canonicalconnection}
Let $(P,\eta,K)$ be an almost para-Hermitian manifold and let $\lc$ be the Levi-Civita connection of $\eta$. Then the connection 
\begin{equation}
\n^c_X = \PP \lc_X \PP + \PPt \lc_X \PPt 
\label{eq:canonicalconnection}
\end{equation}
is called the \textbf{canonical para-Hermitian connection}.
\end{Def}
This canonical connection has all the desired properties. First, we show that it is compatible with the para-Hermitian geometry. The reason we call it ``canonical'' will be justified in the next section.
\begin{Lem}
The canonical connection $\n^c$ preserves $\eta$ and $K$, i.e. $\n^c\eta=\n^cK=0$. Consequently, it also preserves $\omega=\eta K$ and the canonical projectors $\PP$ and $\PPt$.
\end{Lem}
\begin{proof}
The compatibility of $\n^c$ with $\eta$ can be checked directly
\bea
\eta(\n_X^c Y,Z) + \eta( Y,\n_X^c Z)
&=& \eta(\lc_X \PP(Y),\PPt(Z)) + \eta( \PP(Y),\lc_X \PPt(Z))\cr
&+& \eta(\lc_X \PPt(Y),\PP(Z)) + \eta( \PPt(Y),\lc_X \PP(Z))\\
&=& X[\eta(\PP(Y),\PPt(Z)) + \eta(\PPt(Y), \PP(Z))]
= X [\eta(Y,Z)].\nonumber
\eea
The compatibility with $K$ given by $K\n^c =\n^c K$ follows trivially from the definition (\ref{eq:canonicalconnection}) and the properties of projectors $\PP,\PPt$ with $K= \PP-\PPt$.
This establishes that the connection is para-Hermitian.
\end{proof}

\begin{Rem}
In \cite{Ivanov:2003ze}, a class of para-Hermitian connections called  ``canonical'' is given. These connections are parametrized by a real parameter $t$ and according to this terminology, the connection \eqref{eq:canonicalconnection} is a canonical para-Hermitian connection with $t=0$. The canonical connection $\n^c$ can also be expressed using its contorsion tensor \eqref{eq:contorsion}, which is given by the derivative of the almost symplectic form $\omega$
\be\label{contc}
\eta(\n_X^c Y,Z)= \eta(\lc_XY,Z) - \frac12  \lc_X\omega(Y,K(Z)). 
\ee
\end{Rem}

Now we will show that the \gld associated to the canonical connection indeed satisfies the Jacobi identity:
\begin{Thm}
\label{lem:canonical}
Let $(\PS,\eta,K)$ be a $L$-para-Hermitian manifold. Then the generalized Lie derivative associated to the canonical connection $\n^c$,
\begin{equation}
\Lb^c_XY := \Lb_X^{\Dbr^c}Y = \Dbr^c_XY - \Dbr^c_YX + \theta_{\Dbr^c}(Y,X),
\label{eq:canonicalderivative}
\end{equation}
where $\Dbr^c_XY = \n^c_{\PP(X)}Y$, satisfies the Jacobi identity: $J^c=0$. We will call $\Lb^c$ the \textbf{canonical generalized Lie derivative}.
\end{Thm}
\begin{proof}
To show that the Jacobiator $J^c$ vanishes, we will use the expression \eqref{eq:jacobiator} and show that all the individual terms vanish. As was discussed earlier, the last two lines vanish identically. We now evaluate the projected torsion $\tau_\pr^c$ of $\n^c$. Since the canonical connection preserves $K$, it is simply related to $\mathring{T}$, the  ordinary torsion of the Levi-Civita connection $\lc$, which vanishes 
\bea
\tau^c_\pr(X,Y) &=& \PP([\Dbr^c_XY-\Dbr^c_YX]-[\PP(X),\PP(Y)])\cr
&=& \PP(\lc_{\PP(X)} \PP(Y) - \lc_{\PP(Y)} \PP(X) - [\PP(X),\PP(Y)])\cr
&=&  \PP (\mathring{T}(P(X),P(Y))=0.
% \cr &=&  \mathring{\tau}_\pr(X,Y) = 0.
\eea
Next, let us evaluate the first two lines in \eqref{eq:jacobiator} for the canonical connection $\n^c$. Once more using \eqref{eq:etaPPtilde} along with the property $\n^c_X\PP(Y)=\PP\lc_X\PP(Y)$, the projected curvature of $\n^c$, $R_\pr^c$, can be written in terms of the curvature $\mathring{R}$ of $\lc$ (using \eqref{eq:curvature}) as
\begin{equation}\label{eq:lemmaproofcurvatures}
\begin{aligned}
R_\pr^c(X,Y,Z,W) &= \mathring{R}_\pr(X,Y,\PP(Z),\PPt(W)) + \mathring{R}_\pr(X,Y,\PPt(Z),\PP(W)) \\
	&= \mathring{R}(\PP(X),\PP(Y),\PP(Z),\PPt(W)) + \mathring{R}(\PP(X),\PP(Y),\PPt(Z),\PP(W)) \, .
\end{aligned}
\end{equation}
In the expression \eqref{eq:jacobiator}, the first two lines involving $R_\pr$ can be written as a cyclic sum over $(X,Y,Z)$. Since all the remaining parts of $J^c$ vanish as shown above, we have 
\begin{equation}
J^c(X,Y,Z,W)=\cyc [R_\pr^c(X,Y,Z,W) - R_\pr^c(W,Z,X,Y)].
\end{equation}
Rewriting both above terms using \eqref{eq:lemmaproofcurvatures}, we can use the symmetries of $\mathring{R}$ to rewrite $\mathring{R}(W,Z,X,Y)=-\mathring{R}(X,Y,Z,W)$ and apply the Bianchi identity,
\begin{equation}
\cyc \mathring{R}(X,Y,Z,W)=0,
\end{equation}
to cancel all the terms and arrive at the result $J^c=0$.
%They can now be expressed as
%\begin{align}
%&\cyc \big[R_\pr^c(X,Y,Z,W) - R_\pr^c(W,Z,X,Y)\big] = \notag\\
%	&= \cyc \big[ \mathring{R}(P(X),P(Y),P(Z),\tilde{P}(W)) 
%					+ \mathring{R}(P(X),P(Y),\tilde{P}(Z),{P}(W)) \\
%	&\hspace{2.5cm}	+ \mathring{R}(P(X),\tilde{P}(Y),{P}(Z),{P}(W))
%					+ \mathring{R}(\tilde{P}(X),P(Y),{P}(Z),{P}(W))	 \big] = 0 \notag
%\end{align}
%where each of the four terms in the cyclic sum vanishes by the Bianchi identity for the curvature of the Levi-Civita connection $\lc$: $\cyc \mathring{R}(X,Y,Z,W)=0$. Therefore we obtain that the curvature component of the Jacobiator simply
%vanishes for the canonical connection. 
%
%As before, the terms involving $\theta_{D^c}(X,Y)$ in the Jacobiator vanish by isotropy of $L=\mathrm{Im}\,P$. This establish that $J^c=0$.  
\end{proof}

\section{Relationship with Generalized Geometry}

In this section we give the precise relationship between the para-Hermitian geometrical framework, described in previous sections, and generalised geometry. Given an $L$-para-Hermitian manifold $\PS$ we design an isomorphism between $T\PS$ and the generalised bundle $\TT{\cal F}$, where $T\mathcal{F}=L$. Under this isomorphism we show that the Dorfman bracket of generalized geometry \eqref{eq:genLieGG} is  mapped onto the canonical \gld \eqref{eq:canonicalderivative}. As a corollary, this gives another proof of our main theorem \ref{lem:canonical}: the fact that the canonical \gld satisfies the Jacobi identity follows now from the well known fact that the Dorfman bracket satisfies the Jacobi identity. More details on this construction, in particular the underlying Lie algebroid structures, will be discussed in \cite{inprogress}.

Let $(\PS,\eta,\omega)$ be a $L$-para-Hermitian manifold such that $T\PS=L\oplus \Lt$. 
This means that  $L$ is  integrable; let $\mathcal{F}$ be the corresponding foliation of $\PS$ such that $T{\cal F}= L$, i.e. $\cal F$ represents the partition of $\PS$ as a set of leaves. If  $\PS$ is a $2d$-dimensional manifold, $\cal F$ is a $d$-dimensional manifold which is bijectively identified with  $\PS$ \cite{Milnor}. Given a point $p \in \PS$ we denote by $M_p\subset \mathcal{F}$ the leaf passing through $p$. The foliation $\cal F$ is simply $\coprod_{[p]} M_p$ where the index space is the leaf space\footnote{Here $L$ is viewed as an infinitesimal diffeomorphism and the quotient refers to the quotient of $\PS$ by its exponentiated action. 
The leaf space  can also be defined as a quotient  $\PS/\sim_{\cal F}$ where $\sim_{\cal F}$ is the equivalence relation given by $p\sim_{\cal F} q$ if there exist a path $\gamma \in \mathcal{F}$, i-e a path within one leaf of $\mathcal{F}$,  connecting $p$ and $q$.} $\PS/L$.
Let us recall  that by definition we can choose local coordinates $(y,\ty)$ on an open set $U$ such that the intersection of any leaf with the open set $U$ is characterised by the condition $\ty=\mathrm{const.}$\cite{Molino}. In these coordinates, the foliation preserving diffeomorphisms are invertible $C^\infty$ maps $(y,\ty)\mapsto (\Phi(y,\ty),\tilde{\Phi}(\ty) )$, mapping leaves onto leaves.

Given a foliation $\cal F$ viewed as a $d$-dimensional manifold, we can define its tangent and cotangent bundles $T\cal F$ and $T^*\cal F$ and we also introduce  the generalised tangent bundle
$\TT {\cal F} := T{\cal F} \oplus T^*{\cal F}$. If this bundle is restricted to any leaf $M_p$, it gives the usual generalised bundle with base $M_p$: $\TT {\cal F}|_{M_p}= \TT M_p$.
In local coordinates the elements of $\TT{\cal F}$ are given by elements of the form $(x,\alpha)= (x^a(y,\ty)\pa_a + \alpha_a(y,\ty )\rd y^a)$.
There is an isomorphism between the rank $2d$ bundles $\TT {\cal F}$ and  $T\PS$, given by
\begin{align}\label{eq:bundle_iso_eta}
\begin{aligned}
\rho:\TT {\cal F}&\rightarrow T\PS\\
(x,\ap)&\mapsto  x+\etah^{-1}(\ap) \coloneqq X \, .
\end{aligned}
\end{align}
This isomorphism defines an $L$-para-Hermitian structure on $T\PS$, such that $\rho(T\mathcal{F})= L$ and $\rho(T^*\mathcal{F}) =\tilde{L}$. The metric $\eta$ and two-form $\omega$ correspond to the duality pairings 
\be
\eta(\rho(x,\alpha), \rho(y,\beta)) = \imath_x \beta +\imath_y\alpha,
\quad
\omega(\rho(x,\alpha), \rho(y,\beta)) = \imath_x \beta -\imath_y\alpha. 
\ee 
The inverse of $\rho$ is given by 
$
 \rho^{-1}(X) = ( P(X), \etah(\tilde{P}(X)) ).
$

Now that we have an isomorphism  $\rho$ between $\TT {\cal F}$ and $T\PS$ it is possible to pull-back the canonical Lie derivative \eqref{eq:canonicalderivative}  defined in the previous section onto $\TT{\cal F}$. 
 That is we define
 \be
 \rho^*(\Lb^c)_{(x,\ap)}(y,\bt):= \Lb^c_{\rho(x,\ap)}\rho(y,\bt).
 \ee
The Dorfman derivative \eqref{eq:genLieGG} on $\TT{\cal F}$ can be expressed using Cartan's formula as 
\begin{align} 
\mathbb{L}_{(x,\ap)}(y,\bt) = ([x,y],\Lie_x\bt-\Lie_y\ap+\rd [\ap(y)]).
\label{eq:DorfmanDerivative}
\end{align}
Here, the Lie bracket, Lie derivative and the differential are defined on the $d$-dimensional foliation $\cal F$. The relationship between these two Lie derivatives is contained in the following proposition:
\begin{Prop}
Let $(\PS,\eta,\omega)$ be a $L$-para-Hermitian manifold, $\mathcal{F}$ denote the integral foliation of $L$ and  $\rho:\TT{\cal F} \to T\PS$ the isomorphism defined in \eqref{eq:bundle_iso_eta}. Then $\rho$ is a morphism between the Dorfman bracket on $\TT\cal{F}$ and the canonical \gld on $T\PS$, i.e.
\begin{equation}
\rho[\mathbb{L}_{(x,\ap)}(y,\bt)]=\Lb^c_{\rho(x,\ap)}\rho(y,\bt).
\end{equation}
\end{Prop}
This proposition retrospectively justifies the denomination ``canonical'' for the connection entering the definition of $\Lb^c$; it is precisely the connection that reproduces the Dorfman bracket on $\TT\mathcal{F}$ via the natural isomorphism \eqref{eq:bundle_iso_eta}.
\begin{proof}
Under the isomorphism $\rho$, the Dorfman bracket \eqref{eq:DorfmanDerivative} translates to a \gld  on $T\PS=L\oplus  \Lt$. For simplicity and given vectors $X,Y \in T{\cal F}$, we introduce the notations $P(X)=x$ and $P(X)=\tx$ and similarly for $Y$. We define
\be
\Lb^\rho_XY:= \rho[\mathbb{L}_{\rho^{-1}(X)}\rho^{-1}(Y)].
\ee
This can be explicitely written as 
\begin{align}\label{eq:dorfman_LL}
\Lb^\rho_XY
%:= \rho[\mathbb{L}_{(x,\ap)}(y,\bt)]
=[x,y]+\etah^{-1}(\Lie_x\hat\eta(\yt)-\Lie_y\hat\eta(\xt)+\rd [\eta(\xt,y)]).
\end{align}
Notice that because $L$ is integrable, $[x,y]\subset L$. The differential $\rd$ is the horizontal differential along $\cal F$ and the corresponding Lie derivative is given by the Cartan formula $\Lie_x=\imath_x \rd+\rd \imath_x$. The differential $\rd$ can be viewed as the Lie algebroid differential corresponding to the anchor $a: \TT{\cal F} \to T{\cal F}=L$.

Denoting the pairing  $\imath_x\alpha$ between forms and vector as $\alpha(x)$, using that $\rd f(x)=x[f]$ and that $x[\alpha(y)]= [{\cal L}_x \alpha](y)+ \alpha([x,y])$, we can rewrite \eqref{eq:dorfman_LL} as 
\begin{align}
\begin{aligned}
\eta(\Lb^\rho_XY,Z) &= \eta(\tilde{z},[x,y]) +[\Lie_x\hat\eta(\yt)](z)-[\Lie_y\hat\eta(\xt)](z)+ z[\eta(\xt,y)]\cr
&= x[\eta(\yt,z)]-y[\eta(\xt,z)] +z[\eta(\xt,y)]\\
&+ \eta(\xt,[y,z])- \eta(\yt,[x,z])+ \eta(\zt,[x,y]).
\end{aligned}
\end{align}
We can rewrite this identity in terms of the projections 
 $L$ and $\Lt$. Using in particular that $\eta(\tilde{x},y) =\eta(\tilde{P}(X), P(Y))$ we have 
\begin{equation}\label{eq:L_canonical_contracted}
\begin{aligned}
\eta(\Lb^\rho_XY,Z) &= \PP(X)\eta(\PPt(Y),\PP(Z)) 
	- \PP(Y)\eta(\PPt(X),\PP(Z)) +  \PP(Z)\eta(\PPt(X),\PP(Y)) \\
	&+ \eta(\PPt(X),[P(Y),P(Z)]) - \eta(\PPt(Y),[\PP(X),\PP(Z)])
		+ \eta(\PPt(Z),[\PP(X),\PP(Y)]) \, .
\end{aligned}
\end{equation}
Now, using the property that both $L$ and $\Lt$ are isotropic (see \eqref{eq:etaPPtilde}), we have 
\begin{equation}
\begin{aligned}
\eta(\PPt(Z),[\PP(X),\PP(Y)])
	&=\eta(Z,\PP([\PP(X),\PP(Y)])) \\
	&=\eta(Z,\PP(\lc_{\PP(X)}\PP(Y)-\lc_{\PP(Y)}\PP(X)))
\end{aligned}
\end{equation}
where we also used the fact that the Levi-Civita connection is torsionless. If we further use its metricity, we get
\begin{equation}
\PP(X)\eta(\PPt(Y),\PP(Z))=\eta(\PPt(\lc_{\PP(X)}\PPt(Y)),Z)+\eta(Y,\PP\lc_{\PP(X)}\PP(Z)).
\end{equation}
Repeating a similar procedure with the remaining terms in \eqref{eq:L_canonical_contracted} and noticing various cancellations, we arrive at the remarkably simple expression
\begin{equation}\label{eq:gld_proof}
\begin{aligned}
\eta(\Lb^\rho_XY,Z) = \eta(\n^c_{\PP(X)}Y-\n^c_{\PP(Y)}X,Z)+\eta(\n^c_{\PP(Z)}X,Y)=\eta(\Lb^c_XY,Z) .
\end{aligned}
\end{equation}
where $\n^c$ is the canonical connection from Proposition \ref{lem:canonical}. This finally proves our main proposition. 

\end{proof}

\section{Conclusion}
In this work we have studied the relation between DFT and GG in the context of para-Hermitian geometry, we have generalised the Dorfman derivative using the canonical connection and showed that it satisfies the Jacobi identity in the case the geometry is $L$-para-Hermitian. We have also shown its equivalence with the generalised Lie derivative associated with the corresponding foliation. This opens up the possibility to extend DFT to the case where $\eta$ is curved.
 
Our works opens up several new avenues. First, we have not included the presence of fluxes in our description of generalised geometry and it is a natural question to wonder how these would enter in the para-Hermitian context. A related question is to explore whether the $L$-integrability is necessary for the construction of the Dorfman derivative, or whether it can be relaxed modulo extra conditions on the admissible vectors. Also, it is natural to wonder whether the geometrical extension we have presented here can be generalised to the exceptionally extended generalised geometries associated with flux compactifications and Exceptional Field Theory. 

Finally, so far we have only described the generalisation of the kinematics of DFT; it is of utmost importance to understand the interplay between the para-Hermitian geometrical setting and the dynamical metric $\HH$ that enters both the string action and the definition of Born geometry. We hope to come back to these issues in the near future.

\section*{Acknowledgments} 
 L.F. and D.S. would like to greatly thank Ruxandra Moraru and Shengda Hu for the continuous mathematical support and endless discussions around the theme of our paper. L.F. would like to thank collaborators Djordje  Minic and Rob Leigh for an ongoing collaborations and many discussions on related subjects which have been truly inspirational. F.J.R. would like to thank the participants of the workshops in Banff (January 2017) and 
Zagreb (June 2017), especially David Berman, Chris Hull, Dan Waldram, Jeong-Hyuck Park, Chris Blair and Charles Strickland-Constable, for many sharp questions that motivated us to present this work in a better form, and additionally Emanuel Malek for a great many conversations, discussions and clarifications. D.S. would like to thank Dylan Butson for additional mathematical guidance and consultations both related and unrelated to the presented matter.

F.J.R. would like to thank the Perimeter Institute where part of this work was carried out for its kind hospitality and support. While a resident at Perimeter Institute, F.J.R. was a PhD student at Queen Mary University of London supported by an STFC studentship. The work of F.J.R. is now supported by DFG grant TRR33 ``The Dark Universe''.
D.S. is currently a PhD student at Perimeter institute and University of Waterloo. His research is supported by NSERC Discovery Grants 378721.

This research was supported in part by Perimeter Institute for Theoretical Physics. Research at Perimeter Institute is supported by the Government of Canada through Innovation, Science and Economic Development Canada and by the Province of Ontario through the Ministry of Research, Innovation and Science.

\appendix
\section{Properties of para-Hermitian connections}
\label{sec:connections_apndx}
\paragraph{Proof of Lemma \ref{firstlemma}}
For the first part of the lemma we use that 
\bea
\lc_X\omega(Y,Z)&=&
X[\omega(Y,Z)] - \omega(\lc_X Y ,Z)-\omega(Y,\lc_XZ)\cr
&=& X[\eta(K(Y),Z)] - \eta(K(\lc_X Y),Z)-\eta(K(Y),\lc_XZ)\cr
&=& \eta(\lc_X K(Y), Z) - \eta(K(\lc_X Y),Z)\cr\label{etaKYZ1}
&=& \eta(\lc_X K(Y), Z) + \eta(\lc_X Y,K(Z)).
\eea
where we have used the definition of the covariant derivative, the definition of $\omega(X,Y)=\eta (K(X),Y)$, the fact that $\lc$ preserves $\eta$ and the compatibility condition $ \eta (K(X),Y)= -\eta (X,K(Y))$.
The last expression clearly vanishes if  $K(Y)=Y$ and $K(Z)=-Z$ which is the case if $Y=P(Y)$ and $Z=\tilde{P}(Z)$. This  proves (\ref{eq:omegaPtP}). Then,
by expanding (\ref{etaKYZ1}) we get 
\bea
\lc_X\omega(Y,Z)
&=& \eta((\lc_X K)(Y), Z) + \eta(K(\lc_XY),Z) + \eta(\lc_X Y,K(Z)) \cr
&=&  \eta((\lc_X K)(Y), Z),
\eea
which proves (\ref{omegaK}). 

By definition the exterior derivative of a two-form is given by 
\bea
\rd\omega(X,Y,Z)&=& 
\sum_{(X,Y,Z)}(X[\omega(Y,Z)] - \omega([X,Y],Z) ).
\eea
Given a connection $\n$ and using the definition of the covariant derivative acting on $\omega$ this can be written as 
\bea
\rd\omega(X,Y,Z)&=& 
\sum_{(X,Y,Z)}\big[\n_X\omega(Y,Z) + \omega(\n_X Y, Z) +\omega(Y, \n_X Z) - \omega([X,Y],Z) \big]\cr
&=& 
\sum_{(X,Y,Z)}\big[\n_X\omega(Y,Z) + \omega(\hat{T}(X, Y), Z)  \big],\label{domega}
\eea
where we have used the definition of the torsion tensor 
\be
\hat{T}(X, Y) = \n_X Y- \n_Y X - [X,Y].
\ee
For the Levi-Civita connection $\lc$ the torsion vanishes and the lemma follows.

\paragraph{Proof of Lemma \ref{lemma:levicivita_cont}}

Given a para-Hermitian connection we evaluate
 \bea
 \n_X \eta (Y,Z) &=& X[\eta (Y,Z)] - \eta(\n_XY,Z) - \eta(Y,\n_X Z).
 \eea
 Subtracting a similar expression for $\lc_X\eta (Y,Z)$ and imposing that 
 $\n_X\eta=\lc_X\eta=0$ we get that ${\Omega}(X,Y,Z)+{\Omega}(X,Z,Y)=0$.
 Given a para-Hermitian connection we evaluate
 \bea
 \n_X \omega (Y,Z) &=& X[\omega (Y,Z)] - \omega(\n_XY,Z) - \omega(Y,\n_X Z)\cr
 &=&X[\omega (Y,Z)] + \eta(\n_XY,K(Z)) - \eta(K(Y),\n_X Z),
 \eea
 where we have used $\omega(X,Y)=\eta(K(X),Y)=-\eta(X,K(Y))$.  Subtracting a similar expression for $\lc_X\omega (Y,Z)$ and using that 
  $\n_X \omega (Y,Z)=0$ we obtain
  \bea
   \lc_X \omega (Y,Z)
   &=&  {\Omega}(X,Y,K(Z)) - \Omega(X,Z, K(Y)).
  \eea
Now inserting $\PP(Y)$ or $\PPt(Y)$ for $Y$ and similarly for $Z$ gives the result of the lemma since $K\PP=+\PP$ while $K\PPt = -\PPt$. It also gives an alternative proof of \eqref{eq:omegaPtP}.

\section{The Jacobiator}
\label{sec:jacobiator}
In this appendix we fill in some of the details for computing the Jacobiator of the generalised Lie derivative for a projected derivative $\Dbr_X:=\n_{\PP(X)}$. Recall that $\n$ is any metric compatible connection and $P: T\PS \to T\PS$ is a projector but the calculation presented here generalises easily to the case where $\PP$ is an arbitrary anchor map. The image of $\PP$ is denoted by $L$. The complementary projector $\PPt=\id-\PP$ has image $\tilde{L}$ and since $\PP\PPt=0$ it is clear that $\tilde{L}=\Ker\PP$. Therefore if $X\in \im\PP$ then $\PP(X)=X$ and $\PPt(X)=0$. Similarly if $X\in\Ker\PP$ then $\PP(X)=0$ and $\PPt(X)=X$.

\subsection{Computing $J^\Dbr$}
From Definition \ref{def:genLie} we have the projected Dorfman derivative (recalling that $\theta_\Dbr(Y,X)=\eta(Y,\Dbr X)$)
\begin{equation}
\Lb^\Dbr_XY := \Dbr_{X}Y - \Dbr_{Y}X + \theta_\Dbr(Y, X)
\end{equation}
which satisfies the Leibniz rule, is compatible with $\eta$ and obeys a normalisation condition (cf. \eqref{eq:genLieproperties}). Under contraction with another vector this reads
\begin{equation}
\eta(\Lb^\Dbr_XY,Z) = \eta(\Dbr_{X}Y,Z) - \eta(\Dbr_{Y}X,Z) + \eta(Y,\Dbr_ZX)
\end{equation}
We also recall the Jacobiator from Definition \ref{def:jacobiator} 
\begin{equation}
J^\Dbr(X,Y,Z,W) := \eta([\Lb^\Dbr_X,\Lb^\Dbr_Y]Z - \Lb^\Dbr_{\Lb^\Dbr_XY}Z,W) \, .
\end{equation}
We now want to expand this expression to check it is indeed tensorial and analyse the various contributions. We start by using the compatibility property to consider the term
\begin{align}
\eta(\Lb_X^\Dbr \Lb_Y^\Dbr Z,W) &= \PP(X)[\eta(\Lb^\Dbr_Y Z,W)] - \eta(\Lb^\Dbr_Y Z,\Lb^\Dbr_X W) \notag\\
	%&= \PP(X)[\eta([\Dbr_Y Z- \Dbr_Z Y],W) + \eta(Z,\Dbr_WY)] \notag\\
	%& - \eta([\Dbr_Y Z - \Dbr_Z Y],[\Dbr_X W - \Dbr_W X]) \notag\\
	%& - \eta(Z,\Dbr_{[\Dbr_X W - \Dbr_W X]}Y) - \eta(W,\Dbr_{[\Dbr_Y Z - \Dbr_Z Y]}X) \notag\\
	%& - \eta(\eta(Z,\Db Y), \eta(W,\Db X) )
	&= \eta(\Dbr_X\Dbr_YZ,W) - \eta(\Dbr_X\Dbr_ZY,W) 
		+ \eta(Z,\Dbr_X\Dbr_WY) \notag\\
	& + \eta(\Dbr_XZ,\Dbr_WY) + \eta(\Dbr_YZ,\Dbr_WX) - \eta(\Dbr_ZY, \Dbr_WX) \notag\\
	& - \eta(Z,\Dbr_{[\Dbr_X W - \Dbr_W X]}Y) - \eta(W,\Dbr_{[\Dbr_Y Z - \Dbr_Z Y]}X) \notag\\
	& - \eta(\eta(Z,\Db Y), \eta(W,\Db X) ) \, ,
\end{align}
where we also used $\PP(X)[\eta(\Dbr_YZ,W)]=\eta(\Dbr_X\Dbr_Y Z,W) + \eta(\Dbr_YZ,\Dbr_X W)$. Now subtracting the same expression with $X$ and $Y$ exchanged yields
\begin{align}
\eta([\Lb_X^\Dbr,\Lb_Y^\Dbr]Z,W) &= \eta([\Dbr_X,\Dbr_Y]Z,W) 
		- \eta(\Dbr_ZY, \Dbr_WX) + \eta(\Dbr_ZX, \Dbr_WY)\notag\\
	&	- \eta(\Dbr_X\Dbr_ZY - \Dbr_Y\Dbr_ZX,W) + \eta(Z,\Dbr_X\Dbr_WY - \Dbr_Y\Dbr_WX) \notag\\
	& - \eta(Z,\Dbr_{[\Dbr_X W - \Dbr_W X]}Y - \Dbr_{[\Dbr_Y W - \Dbr_W Y]}X) \notag\\
	&	- \eta(W,\Dbr_{[\Dbr_Y Z - \Dbr_Z Y]}X - \Dbr_{[\Dbr_X Z - \Dbr_Z X]}Y) \notag\\
	& - \eta(\eta(Z,\Db Y), \eta(W,\Db X) ) + \eta(\eta(Z,\Db X), \eta(W,\Db Y) ) \, .
\end{align}
On the other hand -- noting that $\eta(\Dbr_UZ,W)=\eta(\eta(W,\Db Z),U)$ -- we find
\begin{align}
\eta(\Lb^\Dbr_{\Lb^\Dbr_XY}Z,W) &= \eta(\Dbr_{\Lb_X^\Dbr Y} Z - \Dbr_Z \Lb_X^\Dbr Y,W) 
		+ \eta(Z,\Dbr_W \Lb_X^\Dbr Y) \notag\\
	&= \eta(\Dbr_{[\Dbr_XY-\Dbr_YX]}Z,W) + \eta(\eta(W, \Db Z), \eta(Y,\Db X)) \notag\\
	&	- \eta(\Dbr_Z \Dbr_XY,W) + \eta(\Dbr_Z \Dbr_YX,W) - \eta(Y,\Dbr_Z\Dbr_W X) \notag\\
	& + \eta(Z, \Dbr_W \Dbr_XY ) - \eta(Z, \Dbr_W \Dbr_YX )+ \eta(Y,\Dbr_W\Dbr_Z X) \notag\\
	& + \eta(\Dbr_WY,\Dbr_ZX) - \eta(\Dbr_ZY,\Dbr_WX)- \eta(Y,\Dbr_{[\Dbr_WZ-\Dbr_ZW]} X) \, .
\end{align}
where we have used $\eta(Z, \Dbr_W \eta(Y,\Db X)) = \eta(\Dbr_WY,\Dbr_Z X) + \eta(Y,\Dbr_W\Dbr_Z X) - \eta(Y,\Dbr_{\Dbr_WZ} X)$.
%\begin{align}
%\eta(Z, \Dbr_W \eta(Y,\Db X)) %&= \PP(W)[\eta(Z,\eta(Y,\Db X))] - \eta(\Dbr_WZ,\eta(Y,\Db X) \notag\\
%	%&= \PP(W)[\eta(Y,\Dbr_Z X)] - \eta(Y,\Dbr_{\Dbr_WZ} X) \notag\\
%	&= \eta(\Dbr_WY,\Dbr_Z X)+ \eta(Y,\Dbr_W\Dbr_Z X)- \eta(Y,\Dbr_{\Dbr_WZ} X) \, .
%\end{align}
Putting everything together and canceling terms, we arrive at the full expression for the Jacobiator $J^\Dbr$
\begin{align}
J^\Dbr(X,Y,Z,W) &=  \eta([\Dbr_X,\Dbr_Y]Z,W) - \eta([\Dbr_X,\Dbr_Z]Y,W) + \eta([\Dbr_Y,\Dbr_Z]X,W) \notag\\
	& + \eta([\Dbr_X,\Dbr_W]Y,Z) - \eta([\Dbr_Y,\Dbr_W]X,Z) + \eta([\Dbr_Z,\Dbr_W]X,Y) \notag\\
	& - \eta(Z,\Dbr_{[\Dbr_X W - \Dbr_W X]}Y) + \eta(Z,\Dbr_{[\Dbr_Y W - \Dbr_W Y]}X) \notag\\
	& - \eta(W,\Dbr_{[\Dbr_Y Z - \Dbr_Z Y]}X) + \eta(W,\Dbr_{[\Dbr_X Z - \Dbr_Z X]}Y) \notag\\
	& + \eta(Y,\Dbr_{[\Dbr_WZ-\Dbr_ZW]} X) - \eta(W,\Dbr_{[\Dbr_XY-\Dbr_YX]}Z) \notag\\
	& - \eta(\eta(Z,\Db Y), \eta(W,\Db X) ) + \eta(\eta(Z,\Db X), \eta(W,\Db Y) ) \notag\\	
	& - \eta(\eta(W, \Db Z), \eta(Y,\Db X)) 	  \, .
\end{align}
In particular we see that all the non-tensorial terms of the form $\eta(\Dbr_XZ,\Dbr_WY)$ cancel out. We can now use the definition of the projected Riemann tensor $R_\pr$, the projected torsion $\tau_\pr$ and the twist of the projected derivative $\theta_\Dbr$ from Definition \ref{def:projectedobjects} 
\begin{equation}
\begin{aligned}
\hat{R}_\pr(X,Y,Z,W)  &:= \eta([\Dbr_X,\Dbr_Y] Z - \Dbr_{[\PP(X),\PP(Y)]}Z, W), \\
\tau_\pr(X,Y) &:= \PP([\Dbr_XY-\Dbr_YX]-[\PP(X),\PP(Y)]), \\
\theta_\Dbr(X,Y) &:= \eta(X,\Db Y).
\end{aligned}
\end{equation}
to express the Jacobiator as
\begin{align}\label{eq:appendix_jacobiator}
J^\Dbr(X,Y,Z,W) &=  R_\pr(X,Y,Z,W) + R_\pr(Y,Z,X,W)  + R_\pr(Z,X,Y,W) \notag\\
	& - R_\pr(W,Z,X,Y) - R_\pr(W,X,Y,Z) - R_\pr(W,Y,Z,X) \notag\\
	& - \eta(W,\n_{\tau_\pr(X,Y)}Z) - \eta(W,\n_{\tau_\pr(Y,Z)}X)  
			- \eta(W,\n_{\tau_\pr(Z,X)}Y) \notag\\
	& - \eta(Z,\n_{\tau_\pr(X,W)}Y) - \eta(X,\n_{\tau_\pr(Y,W)}Z) 
			+ \eta(Y,\n_{\tau_\pr(W,Z)} X) \notag\\
	& + \eta(\theta_\Dbr(Z,X),\theta_\Dbr(W,Y)) - \eta(\theta_\Dbr(Z,Y),\theta_\Dbr(W,X)) \notag\\
	& - \eta(\theta_\Dbr(Y,X),\theta_\Dbr(W,Z)) \, .
\end{align}

\subsection{Anchoring and Tensorality of $J^\Dbr$}
The Jacobiator (Definition \ref{def:jacobiator}) is in general a differential operator just as the \gld $\Lb$ and so it is not tensorial. If we wish to demand that the Jacobiator vanishes, the tensoriality is of course a necessary condition. In the following, we show that the anchoring property
\begin{equation}
\PP(\Lb_XY)=[\PP(X),\PP(Y)] \label{eq:anchoring}
\end{equation}
is another necessary condition and in fact is equivalent to the tensoriality of the Jacobiator. First, we note the following property
\begin{equation}
\Lb^\Dbr_XY+\Lb^\Dbr_YX=\Dbr\eta(X,Y),\label{eq:symmetrization}
\end{equation}
which can be shown by expanding $\Lb^\Dbr_{X+Y}(X+Y)$ and using the normalisation condition $\eta( Y, \Lb^\Dbr_XX)  = \frac12 \PP(Y)[\eta(X,X)]$ along with the symmetry of $\eta$. Now, let us show that the $J^\Dbr$ vanishes only if the \eqref{eq:anchoring} is satisfied.
\begin{Lem}
Let $\Lb$ be a projected generalized Lie derivative\footnote{In the following, we will denote $\Lb^\Dbr$ simply as $\Lb$ to avoid unnecessary cluttering of the expressions.} (Definition \ref{def:genLie}). Then the anchoring property \eqref{eq:anchoring} is a necessary condition for the equation $J^\Dbr=0$ to hold.
\end{Lem}
\begin{proof}
Assume $J^\Dbr=0$. Rewriting this using \eqref{eq:symmetrization} and the definition of the Courant bracket, we get
\begin{equation}
[\Lb_X,\Lb_Y]Z=\Lb_{ \bl X,Y\br^\Dbr}Z+\frac{1}{2}\Lb_{\Dbr \eta(X,Y)}Z=0.
\end{equation}
The symmetric part of this equation is
$\Lb_{\Dbr \eta(X,Y)}=0$, which consequently implies $\Lb_{\Dbr f} X=0$ for all $X$ and any function $f$, because we can always write any $f\in C^\infty(M)$ as $f=\eta(Y,Z)$ for some vector fields $Y$ and $Z$. Using this and again \eqref{eq:symmetrization}, we get
\begin{equation}
\Lb_X\Dbr f+\Lb_{\Dbr f}X=\Dbr\eta(X,\Dbr f)=\Dbr(\PP(X)[f])\Rightarrow \Lb_X\Dbr f=\Dbr(\PP(X)[f]).
\end{equation}
Finally, the statement of the lemma follows from rewriting the metric compatibility of $\Lb$ and using the above derived properties:
\begin{align}
\begin{aligned}
\PP(X)\eta(Y,\Dbr f)&=\eta(\Lb_XY,\Dbr f)+\eta(Y,\Lb_X\Dbr f) \Leftrightarrow\\
\PP(X)[\PP(Y)[f]]&=\PP(\Lb_XY)[f]+\eta(Y,\Dbr(\PP(X)[f])) \Leftrightarrow\\
[\PP(X),\PP(Y)][f]&=\PP(\Lb_XY)[f],
\end{aligned}
\end{align}
since $\eta(Y,\Dbr(\PP(X)[f]))=\PP(Y)[\PP(X)[f]]$. This finishes the proof.
\end{proof}
We will now show that the anchoring property is equivalent to the Jacobiator being tensorial:
\begin{Lem}
Let $\Lb$ be a projected generalized Lie derivative (Definition \ref{def:genLie}). Then the Jacobiator $J^\Dbr$ is tensorial if and only if the anchoring property \eqref{eq:anchoring} holds.
\end{Lem}
\begin{proof}
First, we note the ``Leibniz-like'' properties of $\Lb$
\begin{align}
\Lb_X(fY)&=f\Lb_XY+\PP(X)[f]Y\\
\Lb_{fX}Y&=f\Lb_XY-\PP(Y)[f]X+\eta(X,Y)\Dbr f.
\end{align}
Using this, we expand
\begin{align}\label{eq:tensoriality1}
\begin{aligned}
\Lb_X\Lb_{fY} Z &= \Lb_X(f\Lb_YZ-\PP(Z)[f]Y + \eta(Y,Z)\Dbr f) \\
	&= f\Lb_X\Lb_YZ + \PP(X)[f]\Lb_YZ - \PP(Z)[f]\Lb_XY \\
	&- \PP(X)\PP(Z)[f]Y + \eta(Y,Z)\Lb_X\Dbr f + \PP(X)\eta(Y,Z)\Dbr f.
\end{aligned}
\end{align}
Similarly, we find out
\begin{align}\label{eq:tensoriality2}
\begin{aligned}
\Lb_{fY}\Lb_{X} Z&=f\Lb_Y\Lb_X Z-\PP(\Lb_XZ)[f] Y+\eta(Y,\Lb_XZ) \Dbr f\\
\Lb_{\Lb_X(fY)}Z&=f\Lb_{\Lb_XY}Z-\PP(Z)[f]\Lb_XY+\eta(\Lb_XY,Z)\Dbr f\\
&+\PP(X)[f]\Lb_YZ-\PP(Z)\PP(X)[f]Y+\eta(Y,Z)\Dbr(\PP(X)[f]).
\end{aligned}
\end{align}
By combining \eqref{eq:tensoriality1} and \eqref{eq:tensoriality2}, and using the metric compatibility of $\Lb$, we get
\begin{align}
([\Lb_X,\Lb_{fY}]-\Lb_{\Lb_X(fY)})Z &= f([\Lb_X,\Lb_{Y}]-\Lb_{\Lb_XY})Z-([\PP(X),\PP(Z)]
	-\PP(\Lb_XZ))[f] \notag\\
&+ \eta(Y,Z)(\Lb_X\Dbr f-\Dbr(\PP(X)[f])).
\end{align}
Therefore, the tensoriality of $J^\Dbr(X,Y,Z,W)$ in the $Y$-entry is ensured if and only if
\begin{align}
[\PP(X),\PP(Z)]-\PP(\Lb_XZ)&=0,\\
\Lb_X\Dbr f-\Dbr[\PP(X)[f]]&=0.
\end{align}
The first equation is exactly the anchoring property \eqref{eq:anchoring} and contracting the second equation with $\eta(W,\cdot)$ and using metric compatibility of $\Lb$ reveals that it is equivalent to the anchoring property as well
\begin{align}
\eta(W,\Lb_X\Dbr f-\Dbr(\PP(X)[f]))&= \PP(X)\eta(W,\Dbr f)-\eta(\Lb_XW,\Dbr f)-\PP(W)\PP(X)[f]\notag\\
&=([\PP(X),\PP(W)]-\PP(\Lb_XW))[f]=0.
\end{align}
This shows that $J^\Dbr$ is tensorial in the $Y$-entry if and only if \eqref{eq:anchoring} holds. The tensoriality in the $X$ and $Z$ entries is shown by analogous calculation while the tensoriality in the $W$ entry follows directly from the definition of $J^\Dbr$.
\end{proof}

We have shown above that the tensoriality is equivalent to \eqref{eq:anchoring}. Moreover, the anchoring itself has some strong implications. First, it implies the that $L=\im \PP$ is an integrable distribution, which follows directly from the Frobenius integrability condition, $[\PP(X),\PP(Y)]\subset L$. Second, it implies that the projected torsion $\tau_\pr(X,Y)$
%\begin{equation}
%\tau_\pr(X,Y) = \PP([\Dbr_XY-\Dbr_YX]-[\PP(X),\PP(Y)])
%\end{equation}
vanishes. This can be seen from the definition of the projected \gld and the fact that the twist $\theta_\Dbr$ is valued in $\Lt=\Ker \PP$
\begin{align}
\tau_\pr(X,Y) &= \PP([\Dbr_XY-\Dbr_YX]-[\PP(X),\PP(Y)]) \notag\\
&= \PP(\Dbr_XY-\Dbr_YX+\theta_\Dbr(Y,X)-[\PP(X),\PP(Y)]) \\
&= \PP(\PP(\Lb_XY)-[\PP(X),\PP(Y)]) \, . \notag
\end{align}
Therefore if we require tensoriality of the Jacobiator and the isotropy of $\Lt$
%and recall that $\theta_\Dbr$ is valued in $\Lt$, which is isotropic w.r.t. $\eta$, 
the expression\eqref{eq:appendix_jacobiator} simplifies to
\begin{align}
\begin{aligned}
J^\Dbr(X,Y,Z,W) &=  R_\pr(X,Y,Z,W) + R_\pr(Y,Z,X,W)  + R_\pr(Z,X,Y,W)\\
	& - R_\pr(W,Z,X,Y) - R_\pr(W,X,Y,Z) - R_\pr(W,Y,Z,X).
	\end{aligned}
\end{align}

\bibliographystyle{JHEP} 
\bibliography{mybib}

\end{document}